\documentclass[12pt]{article}

\usepackage[colorlinks, linkcolor=blue, citecolor=blue]{hyperref}           
\usepackage{color}
\usepackage{graphicx,subfigure,amsmath,amssymb,amsfonts,bm,epsfig,epsf,url,dsfont}
\usepackage{amsthm}
\usepackage{tikz}
\usepackage{bbm}      
\usepackage{booktabs}
\usepackage{cases}
\usepackage{fullpage}
\usepackage[small,bf]{caption}

\usepackage{overpic}

\usepackage[top=1in,bottom=1in,left=1in,right=1in]{geometry}
\usepackage{fancybox}
\usepackage{hyperref}
\usepackage{algorithm}
\usepackage{verbatim}
\usepackage{algorithmic}

\usepackage{mathtools}

\usepackage{scalerel,stackengine}
\stackMath
\newcommand\reallywidehat[1]{%
\savestack{\tmpbox}{\stretchto{%
  \scaleto{%
    \scalerel*[\widthof{\ensuremath{#1}}]{\kern-.6pt\bigwedge\kern-.6pt}%
    {\rule[-\textheight/2]{1ex}{\textheight}}%WIDTH-LIMITED BIG WEDGE
  }{\textheight}% 
}{0.5ex}}%
\stackon[1pt]{#1}{\tmpbox}%
}

\usepackage{rotating}

\newcommand{\bE}{\mathbb{E}}

\newtheorem{definition}{Definition}

\newtheorem{theorem}{Theorem}

\newtheorem{lemma}{Lemma}
\newtheorem{prop}{Proposition}

\newenvironment{fminipage}%
  {\begin{Sbox}\begin{minipage}}%
  {\end{minipage}\end{Sbox}\fbox{\TheSbox}}

\makeatletter
\newcommand*{\rom}[1]{\expandafter\@slowromancap\romannumeral #1@}
\makeatother
\newtheorem*{remark}{Remark}

\newcommand{\Ind}{\mathbbm{1}}

\newcommand{\abs}[1]{\left|#1\right|}
\newcommand{\R}{\mathbb{R}} 
\newcommand{\N}{\mathbb{N}}

\newcommand{\E}{\mathbb{E}}
\newcommand{\s}[1]{\mathsf{#1}}

\newcommand{\calp}{\mathcal{p}}
\def\P{{\mathbb P}}
\def\S{{\mathbb S}}

\newcommand {\pr} {\mathbb{P}}

\newcommand{\calA}{{\cal A}}

\newcommand{\calH}{{\cal H}}
\newcommand{\calI}{{\cal I}}

\newcommand{\calK}{{\cal K}}
\newcommand{\calL}{{\cal L}}

\newcommand{\calN}{{\cal N}}

\newcommand{\calP}{{\cal P}}
\newcommand{\calQ}{{\cal Q}}

\newcommand{\be}{\begin{equation}}
\newcommand{\ee}{\end{equation}}
\newcommand{\beqna}{\begin{eqnarray}}
\newcommand{\eeqna}{\end{eqnarray}}

\DeclarePairedDelimiterX{\set}[1]{\{}{\}}{\setargs{#1}}
\DeclarePairedDelimiterX{\cond}[1]{[}{]}{\setargs{#1}}
\NewDocumentCommand{\setargs}{>{\SplitArgument{1}{;}}m}
{\setargsaux#1}
\NewDocumentCommand{\setargsaux}{mm}
{\IfNoValueTF{#2}{#1} {#1\,\delimsize|\,\mathopen{}#2}}%{#1\:;\:#2}

\DeclarePairedDelimiter\parenv{\lparen}{\rparen}
\newcommand{\eqdef}{\triangleq}

\newcommand{\indep}{\perp \!\!\! \perp}

\newcommand{\p}[1]{\left(#1\right)}
\newcommand{\pp}[1]{\left[#1\right]}
\newcommand{\ppp}[1]{\left\{#1\right\}}
\newcommand{\norm}[1]{\left\|#1\right\|}
\newcommand{\innerP}[1]{\left\langle#1\right\rangle}

\usepackage{xspace}
\usepackage[scr=esstix,cal=boondox]{mathalfa} 

\allowdisplaybreaks
\begin{document}

\title{Detection of Correlated Random Vectors}

\author{Dor~Elimelech\thanks{D. Elimelech is with the School of Electrical and Computer Engineering at Ben-Gurion university, {B}eer {S}heva 84105, Israel (e-mail:  \texttt{doreli@post.bgu.ac.il}). The work of D.Elimelech was supported by the ISRAEL SCIENCE FOUNDATION (grant No.  985/23).}~~~~~~~~~~~Wasim~Huleihel\thanks{W. Huleihel is with the Department of Electrical Engineering-Systems at Tel Aviv university, {T}el {A}viv 6997801, Israel (e-mail:  \texttt{wasimh@tauex.tau.ac.il}). The work of W. Huleihel was supported by the ISRAEL SCIENCE FOUNDATION (grant No. 1734/21).}}

\maketitle

\begin{abstract}

In this paper, we investigate the problem of deciding whether two standard normal random vectors $\s{X}\in\mathbb{R}^{n}$ and $\s{Y}\in\mathbb{R}^{n}$ are correlated or not. This is formulated as a hypothesis testing problem, where under the null hypothesis, these vectors are statistically independent, while under the alternative, $\s{X}$ and a randomly and uniformly permuted version of $\s{Y}$, are correlated with correlation $\rho$. We analyze the thresholds at which optimal testing is information-theoretically impossible and possible, as a function of $n$ and $\rho$. To derive our information-theoretic lower bounds, we develop a novel technique for evaluating the second moment of the likelihood ratio using an orthogonal polynomials expansion, which among other things, reveals a surprising connection to integer partition functions. We also study a multi-dimensional generalization of the above setting, where rather than two vectors we observe two databases/matrices, and furthermore allow for partial correlations between these two.
   
\end{abstract}

\section{Introduction}

Consider the following binary hypothesis testing problem. Under the null hypothesis, two $n$-dimensional standard normal random vectors $\s{X} = (X_1,\ldots,X_n)$ and $\s{Y} = (Y_1,\ldots,Y_n)$ are drawn \emph{independently} at random. Under the alternative hypothesis, the entries of $\s{X}$ are \emph{correlated} with a randomly and uniformly permuted version of the entries of $\s{Y}$ (see, Fig.~\ref{fig:comp} for an illustration). Then, under what conditions, one can infer/decide whether $\s{X}$ and $\s{Y}$ are correlated or not?

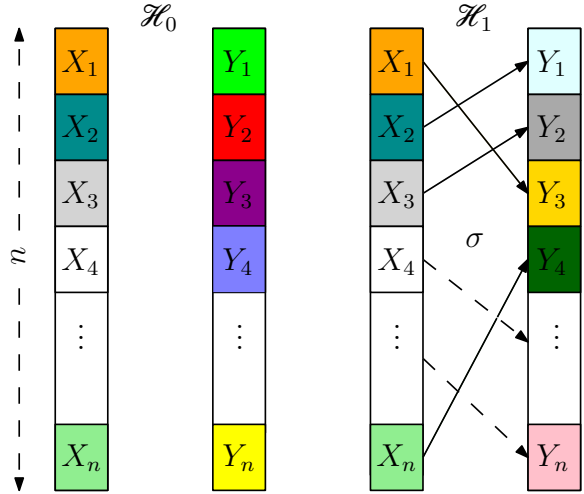
\begin{figure}[ht]
\centering
\begin{overpic}[scale=0.3]
    {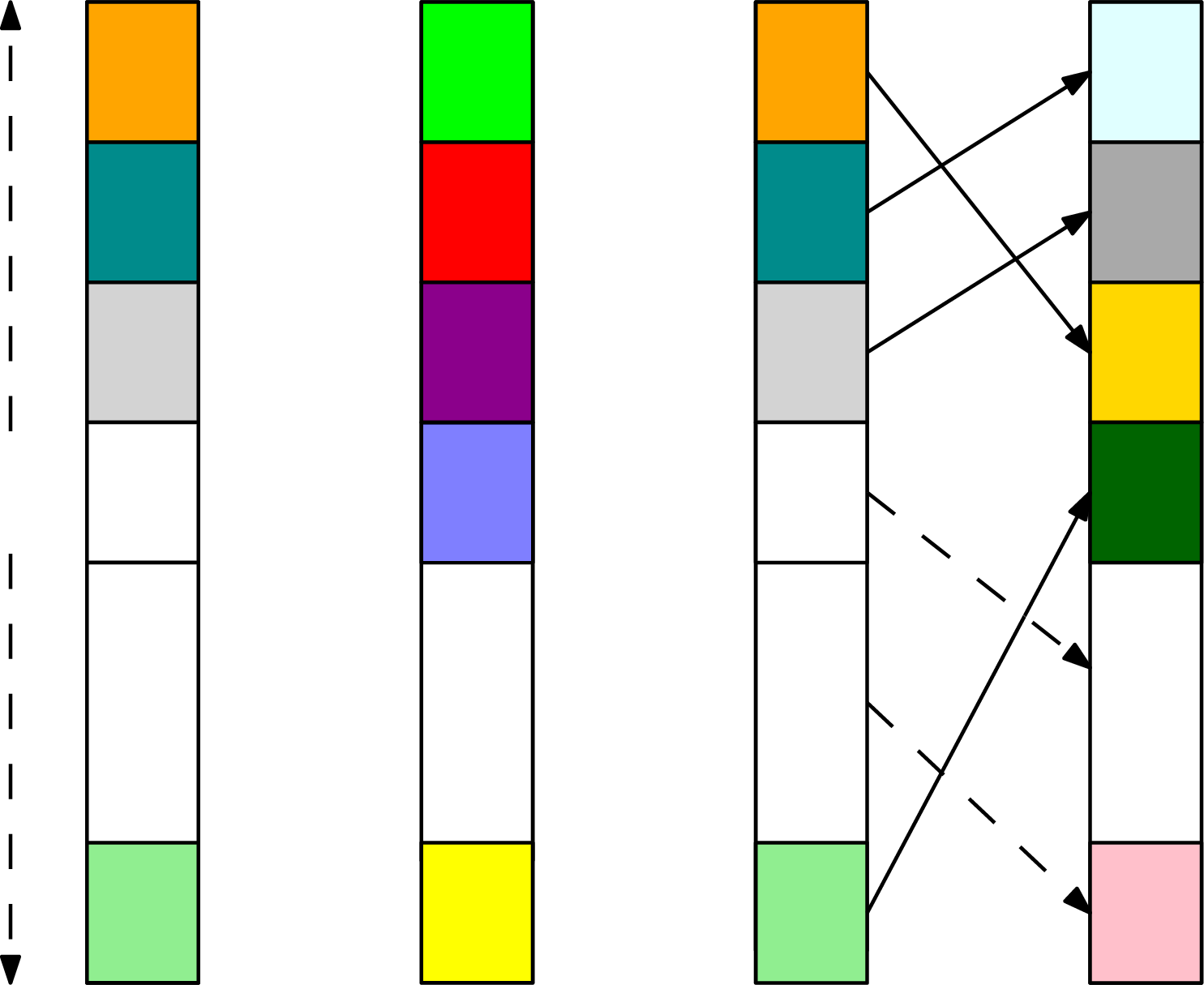}
    \put(22.2,82){$\calH_0$} 
    \put(77.7,82){$\calH_1$} 
    
    \put(-0.5,39.8){\begin{turn}{90}$n$\end{turn}}
    \put(8.3,74){$X_1$} 
    \put(8.3,62.5){$X_2$} 
    \put(8.3,50.8){$X_3$} 
    \put(8.3,39){$X_4$} 
    \put(11,25){$\vdots$} 
    \put(8.3,4.4){$X_n$} 
    \put(79.5,43){$\sigma$} 

    \put(64,74){$X_1$} 
    \put(64,62.5){$X_2$} 
    \put(64,50.8){$X_3$} 
    \put(64,39){$X_4$} 
    \put(66.7,25){$\vdots$} 
    \put(63.8,4.4){$X_n$} 

    \put(36.7,74){$Y_1$} 
    \put(36.7,62.5){$Y_2$} 
    \put(36.7,50.8){$Y_3$} 
    \put(36.7,39){$Y_4$} 
    \put(39,25){$\vdots$} 
    \put(36.7,4.4){$Y_n$} 

    \put(92,74){$Y_1$} 
    \put(92,62.5){$Y_2$} 
    \put(92,50.8){$Y_3$} 
    \put(92,39){$Y_4$} 
    \put(94.8,25){$\vdots$} 
    \put(92,4.4){$Y_n$} 
    
\end{overpic}
\caption{
An illustration of the detection problem. On the left are the uncorrelated vectors under the null hypothesis $\calH_0$. On the right are the vectors $\s{X}$ and $\s{Y}$ under the hypothesis $\calH_1$, where correlated elements are marked with a similar color.
}
\label{fig:comp}
\end{figure}

The above setting is motivated by a recent line of work on what is known as the \emph{data alignment problem}, introduced and explored in, e.g., \cite{10.1109/ISIT.2018.8437908,pmlr-v89-dai19b}. In this problem, $\s{X}$ and $\s{Y}$ are $\mathbb{R}^{n\times d}$ random matrices, with $n$ typically designating the number of users, each with $d$ features. As in the problem above, there is an unknown (or randomly generated) permutation that matches users in $\s{X}$ with those in $\s{Y}$. When a pair of entries from these databases is matched, their features are dependent according to a known distribution, whereas for unmatched entries, the features are independent. In the recovery problem, the goal is to \emph{recover} the unknown permutation. When the two databases have independent standard normal entries, with correlation coefficient $\rho$ between the entries of matched rows, it has been demonstrated in \cite{pmlr-v89-dai19b} that perfect recovery is attainable if $\rho^2 = 1-o(n^{-4/d})$, while it becomes impossible if $\rho^2 = 1-\omega(n^{-4/d})$ as both $n$ and $d$ tend to infinity. 

The \emph{detection} problem, which is more relevant to our paper, has also undergone extensive investigation in \cite{9834731,nazer2022detecting,tamir,HuleihelElimelech}, in the Gaussian case, and recently, for general distributions in \cite{paslev2023testing}. In a nutshell, in the regime where $d\to\infty$, if $\rho^2d\to0$ weak detection (performing slightly better than random guessing) is information-theoretically impossible, independently of the value of $n$, while if $\rho^2d\to\infty$, efficient strong detection (with vanishing error probability) is possible. While the statistical limits for detection and recovery are clear when $d\to\infty$, the case where $d$ is fixed is still a mystery. Specifically, for $d\geq d_0$, and some $d_0\in\mathbb{N}$, it was shown in \cite{HuleihelElimelech} that if $\rho^2 = 1-o(n^{-2/(d-1)})$ then strong detection is possible, while if $\rho\leq\rho^\star(d)$, for some function $\rho^\star(d)$ of $d$, then strong detection is impossible. In particular, for $d=1$ we have $\rho^\star(1)=1/2$, while an algorithmic upper bound in this case is missing. Accordingly, this evident substantial gap between the current known upper and lower bounds, even in the elementary case of $d=1$, sets the main goal of our paper. 

The main results of our work are sharp information-theoretic thresholds for the impossibility and possibility of strong detection in the one-dimensional case, thereby closing the aforementioned gap in \cite{HuleihelElimelech,nazer2022detecting,paslev2023testing,tamir}. Specifically, we prove that strong detection is impossible when $\rho^2$ is bounded away from $1$ and possible when $\rho^2=1-o(n^{-4})$. Another significant aspect of our results concerns with the proof technique of our lower bound, in which we abandon the standard approach of upper bounding the second moment of the likelihood ratio directly. Instead, we use an orthogonal decomposition of the likelihood ratio with respect to (w.r.t.) Hermite polynomials, and compute the second moment using Parseval's identity. Not only that this new approach enables us to improve upon the $\rho^{\star}$ threshold from \cite{HuleihelElimelech}, and close the gap, it also reveals an intriguing connection between random permutations and integer partition functions.

Then, in the second part of our paper, we use the techniques established for the one-dimensional, and study a generalized version, where $\s{X}$ and $\s{Y}$ are now $n\times d$ databases/matrices, and further consider the scenario of partial correlation. To wit, under the alternative hypothesis, only a random subset of $k$ out of the $n$ rows of $\s{X}$ are correlated with some $k$ rows of a randomly permuted version of $\s{Y}$. This aims to model the natural and practical scenario where in two different databases (e.g., movie rating datasets) only a subset of the total population of users in each database, is common. For this model, we prove non-trivial thresholds for impossibility and possibility of weak and strong detection w.r.t. the parameters $\rho,d,n$ and $k$. For example, in the fully correlated case where $k=n$, our results coincide with the sharp thresholds proved in \cite{HuleihelElimelech}. When $k=O(\log n)$, our results suggest that detection becomes statistically harder compared to the fully correlated model, regardless of the value of $d$.

We will now present a brief overview of a few related works. The problem of partially recovering the hidden alignment was investigated in \cite{9174507}. In \cite{ShiraniISIT} necessary and sufficient conditions for successful recovery using a typicality-based framework were developed. Furthermore, \cite{Bakirtas2021DatabaseMU} and \cite{Bakirtas2022DatabaseMU} addressed alignment recovery in scenarios involving feature deletions and repetitions, respectively. Finally, the task of joint detection and recovery was proposed and analyzed in \cite{tamir}. We also mention a series of papers \cite{riederer2016linking,8993856,9354190,unnikrishnan2014asymptotically}, in which a variant the database alignment problem (i.e., recovery) was considered. Specifically, these papers deal with the task of matching random sequences to their \textit{distinct} generative distributions, which are either known or unknown but a training sequence from each database is available. While this is similar in nature to the database alignment problem, the techniques and results are quite different. It is noteworthy that the difficulties in database alignment and detection are closely linked to various planted matching problems, particularly the graph alignment problem. This problem involves detecting edge correlations between two random graphs with unlabeled nodes, see, e.g., \cite{GraphAl1,GraphAl2,GraphAl3,GraphAl4,wu2020testing,GraphAl5,GraphAl6,GraphAl7}.

\section{Model Formulation}

\noindent\textbf{Preliminaries and notation.} Throughout this work, we use lower case letters $x$ for scalars, upper case letters $X$ for random variables. We denote the set of integer $\ppp{1,2,\dots n}$ by $[n]$, and with abuse of notation, we use $\binom{[n]}{k}$ in order to denote the set of all subsets of $[n]$ of size $k$. We denote the set of all permutation over $[n]$ by $\mathbb{S}_n$, and for $\sigma\in\mathbb{S}_n$, we let $\sigma(i)$ denote the value to which $\sigma$ maps $i\in[n]$.  For two functions $f,g:\N\to \R$ we say that $f=o(g)$ if $f(n)/g(n)\xrightarrow[]{n\to\infty}0$, and $f=O(g)$ if there exists a constant $C$ such that $f(n)\leq C\cdot g(n)$ for all $n\in \N$. By the same convention, we say that $f=\omega(g)$ if $g=o(f)$ and $f=\Omega(g)$ if $g=O(f)$. For two measures $\mu$ and $\nu$ on the same measurable space such that $\nu$ is absolutely continuous w.r.t. $\mu$ (that is, $\mu(A)=0\implies \nu(A)=0$ for all $A$), we denote the Radon-Nikodym derivative of $\nu$ w.r.t. $\mu$ by $\frac{\mathrm{d}\mu}{\mathrm{d}\nu}$. We use $\calN(\mathbf{\mu},\Sigma)$ to denote the normal distribution with expectation $\mathbf{\mu}$ and covariance matrix $\Sigma$. For $d\in \N$ we use $\calN(\mathbf{\mu},\Sigma)^{\otimes d}$ to denote the independent product of $d$ $\calN(\mathbf{\mu},\Sigma)$-Normal distributions. 

\vspace{0.2cm}
\noindent\textbf{Probabilistic model.} As described in the introduction, we deal with the following hypothesis testing problem, 
%Under the null hypothesis $\calH_0$, the Gaussian databases $X^n$ and $Y^n$ are generated independently with $X_1,\ldots,X_n,Y_1,\ldots,Y_n\sim N(0_d,\mathbf{I}_d)$. Let $\mathbb{P}_0$ denote the resulting distribution over $(\s{X},\s{Y})$. Under the alternative hypothesis $\calH_1$,  the rows of $X^n$ are correlated with some an unknown permuted version of the rows of $Y^n$. Specifically, let $\sigma$ be an unknown permutation of $[n]$. Then, for each $i\in[n]$, under $\calH_1$, the random vectors $X_i$ and $Y_{\sigma(i)}$ are jointly Gaussian with correlation $\rho\neq0$. To conclude, we deal with the following hypothesis testing problem,
%\begin{equation}
%\begin{aligned}\label{eqn:decproblem}
%    &\calH_0: (X_1,Y_1),\ldots,(X_n,Y_n)\stackrel{\mathrm{i.i.d}}{\sim} \calN^{\otimes d}(\mathbf{0},\mathbf{I}_{2\times 2})\\
%& \calH_1:
%(X_1,Y_{\sigma(1)}),\ldots,(X_{\sigma(n)},Y_{\sigma(n)}) \overset{\mathrm{i.i.d}}{\sim} \calN^{\otimes d}(\mathbf{0},\Sigma_\rho),
%\end{aligned}
%\end{equation}
\begin{equation}
\begin{aligned}\label{eqn:decproblem}
    &\calH_0: (X_1,Y_1),\ldots,(X_n,Y_n)\stackrel{\mathrm{i.i.d}}{\sim} \calN(\mathbf{0},\Sigma_0)\\
& \calH_1:
(X_1,Y_{\sigma(1)}),\ldots,(X_{\sigma(n)},Y_{\sigma(n)}) \overset{\mathrm{i.i.d}}{\sim} \calN(\mathbf{0},\Sigma_\rho),
\end{aligned}
\end{equation}
where $\sigma\sim\s{Unif}(\mathbb{S}_n)$, and for $|\rho|\leq1$,
\begin{align}
    \Sigma_\rho\triangleq\begin{bmatrix}
1 & \rho\\
\rho & 1
\end{bmatrix}.
\end{align}
Given $\sigma\in \S_n$, we denote the joint distribution measure of $(\s{X},\s{Y})$  under the hypothesis $
\calH_{1}$ by $\P_{\calH_1\vert\sigma}$, and under the hypothesis $
\calH_{0}$ by $\P_{\calH_0}$.

\vspace{0.2cm}
\noindent\textbf{Problem formulation.} A test function for our problem is a function $\phi:\R^{n}\times\R^{n}\to \{0,1\}$, designed to determine which of the hypothesis $\calH_0,\calH_1$ occurred. The risk of a test $\phi$ is defined as the sum of its average Type-I and Type-II error probabilities, i.e.,
\begin{align}
\s{R}(\phi)\triangleq \P_{\calH_0}[\phi(\s{X},\s{Y})=1]+\P_{\calH_1}[\phi(\s{X},\s{Y})=0],
\end{align}
with $\P_{\calH_1} = \bE_{\sigma\sim\s{Unif}(\mathbb{S}_n)}[\P_{\calH_1\vert\sigma}]$. The Bayesian risk associated with our hypothesis detection problem is 
\begin{align}
\s{R}^\star\triangleq\inf_{\phi:\R^{n}\times\R^{n}\to \{0,1\}}\s{R}(\phi).
\end{align}
We remark that $\s{R}$ is a function of $\rho$ and $n$, however, we omit them from our notation for the benefit of readability.  

\begin{definition}
A sequence $(\rho,n)=(\rho_k,n_k)_k$ is said to be:
\begin{enumerate}
    \item Admissible for strong detection if 
    % there exists an algorithm $\phi$ such that 
    $\lim_{k\to \infty}\s{R}^\star=0$.
    \item Admissible for weak detection if
    % if there exists an algorithm $\phi$ such that 
    $\limsup_{k\to \infty} \s{R}^\star<1$.
\end{enumerate}
Clearly, admissibility of strong detection implies the admissibility of weak detection.
\end{definition}
While admissibility of strong detection clearly refers to the existence of algorithms that correctly detects with probability that tends to $1$, weak detection implies the the existence of algorithms which are asymptotically better then randomly guessing which of the hypothesis occurred. The next two sections are devoted to the investigation of the one-dimensional detection problem in \eqref{eqn:decproblem}. As mentioned in the Introduction, we investigate also a generalization of \eqref{eqn:decproblem} to high-dimensions and partial correlations (see, Section~\ref{sec:extension}). 

\section{Phase Transition in 1D}\label{sec:1D}

In this section, we present our main results starting with lower bounds. 

\vspace{0.2cm}
\noindent\textbf{Lower bounds.} We have the following impossibility guarantees for strong and weak detection, respectively. 
\begin{theorem}[Impossibility of strong detection]\label{th:lowerStrong}
Consider the detection problem in \eqref{eqn:decproblem}.
For any sequence $(\rho,n)=(\rho_k,n_k)_k$ such that $\rho^2=1-\Omega(1)$, we have 
\begin{align}
    \s{R}^\star=\Omega(1).
\end{align}
Namely, strong detection is impossible.
\end{theorem}
\begin{theorem}[Impossibility of weak detection]\label{th:lowerWeak}
Consider the detection problem in \eqref{eqn:decproblem}. For any sequence $(\rho,n)=(\rho_k,n_k)_k$ such that $\rho^2=o(1)$, we have 
\begin{align}
    \s{R}^\star=1-o(1).
\end{align}
Namely, weak detection is impossible.
\end{theorem}
Let us discuss the above results. From Theorem~\ref{th:lowerStrong} we see that strong detection is statistically impossible if $\rho^2$ is bounded away from $1,$ and in particular, for any \emph{fixed} $\rho^2<1$. This improves \emph{strictly} the known bound in \cite{HuleihelElimelech}, where it was shown that strong detection is statistically impossible whenever $\rho^2<1/2$. As for weak detection, our lower bound in Theorem~\ref{th:lowerWeak} implies that weak detection is statistically impossible whenever the correlation $\rho^2$ decays to zero; this matches the results in \cite{HuleihelElimelech}.
%Below we will show that strong detection is possible if $\rho^2$ approaches unity sufficiently fast. \black{The lower bound of Theorem~\ref{th:lowerWeak}, which matches the one of \cite{HuleihelElimelech}, shows  that weak detection is impossible whenever the correlation $\rho^2$ decays to zero. Below we show that weak detection is possible when $\rho^2$ is bounded away from zero, which also improves upon the existing results, which does not apply to the one dimensional case (as the threshold of \cite{HuleihelElimelech} is non-trivial only when $d\geq 60\cdot\log2$). Perhaps unsurprisingly, the above results show that weak detection is statistically easier than strong detection as the condition on the correlation for strong detection is (i.e., $\rho^2=\Omega(1)$) is weaker than the condition for strong detection (i.e, $\rho^2=1-o(1)$).}

\vspace{0.2cm}
\noindent\textbf{Upper bounds.} Next, we present our upper bounds, starting with a strong detection guarantee. To that end, we now define a testing procedure and analyze its performance. For simplicity of notation, we denote $Q_{XY} = \calN(\mathbf{0},\Sigma_0)$ and $P_{XY} = \calN(\mathbf{0},\Sigma_\rho)$, and the associated joint probability density functions by $f_{P}$ and $f_{Q}$, respectively. Then, for $x,y\in\mathbb{R}$, let the individual likelihood of $x$ and $y$ be defined by
\begin{align}
    \calL_{\s{I}}&(x,y)\triangleq \log\frac{f_{P}(x,y)}{f_Q(x,y)}\\
    &=-\frac{1}{2}\log(1-\rho^2)-\frac{\rho^2}{2(1-\rho^2)}(x^2+y^2) +\frac{\rho}{1-\rho^2}xy,
\end{align}
and define
\begin{align}
    \phi_{\s{count}}(\s{X},\s{Y})\triangleq\Ind\ppp{\sum_{i,j=1}^n\calI(X_i,Y_j)\geq \frac{1}{2} n\calP_{\rho}},\label{eqn:testcount}
\end{align}
where $\calI(X_i,Y_j)\triangleq\Ind\ppp{\calL_{\s{I}}(X_{i},Y_{j})\geq\tau_{\s{count}}}$,  $\calP_{\rho}\triangleq\pr_{P}\pp{\calL_{\s{I}}(A,B)\geq \tau_{\s{count}}}$, and $\tau_{\s{count}}\in\mathbb{R}$, where $\pr_{P}(\mathcal{F})$ means that the probability of the event $\mathcal{F}$ is evaluated w.r.t. the joint distribution $P_{XY}$. Similarly, we also define $\calQ_{\rho}\triangleq\pr_{Q}\pp{\calL_{\s{I}}(A,B)\geq \tau_{\s{count}}}$, where the probability is evaluated w.r.t. the joint distribution $Q_{XY}$.
%\begin{align}\label{eqn:p_rho}
%    \pr_\rho \triangleq \mathcal{N}\p{\begin{bmatrix}
%0 \\
%0 
%\end{bmatrix},\begin{bmatrix}
%1 & \rho\\
%\rho & 1
%\end{bmatrix}}.
%\end{align}
\sloppy
Roughly speaking, $\phi_{\s{count}}$ counts the number of pairs whose likelihood individually exceed the threshold $\frac{1}{2} n\calP_{\rho}$. Intuitively, since the individual likelihoods are proportional to the products of the entries of $\s{X}$ and $\s{Y}$, we expect the count to be larger under the alternative hypothesis, when these vectors are correlated. While the choice of the threshold $\frac{1}{2}n\calP_\rho$ is not unique (and perhaps even not optimal), it allows us to control the Type-I and Type-II error probabilities. Indeed, intuitively, when $\rho^2$ is sufficiently close to unity (which as will be seen below is the relevant regime where the count test is successful) we get that $\calQ_\rho=o(n^2\calP_\rho)$, and therefore, the summation of $\calL_{\s{I}}(X_i,Y_i)$ over uncorrelated pairs $\{(X_i,Y_j)\}_{i,j}$ is ``relatively small" with high probability (e.g., by the law of large numbers); this eliminates errors of Type-I. Moreover, we have that $\calP_\rho$ is strictly positive, and thus the summation of $\calL_{\s{I}}(X_i,Y_i)$ over correlated pairs $\{(X_i,Y_j)\}_{i,j}$ is strictly positive (again, by the law of large numbers), which eliminates errors of Type-II. It should be mentioned here that a similar test was proposed in \cite{HuleihelElimelech}. Specifically, in \cite{HuleihelElimelech}, the number of \emph{normalized} inner products is being counted, while here we count the inner products directly. As it turns out, the normalization procedure in \cite{HuleihelElimelech} excludes the $d=1$ case we consider in this section.

To present our main result, we need the following definitions. For $\theta\in(-d_{\s{KL}}(Q_{XY}||P_{XY}),d_{\s{KL}}(P_{XY}||Q_{XY}))$, we define the Chernoff's exponents $E_{P},E_Q:\mathbb{R}\to[-\infty,\infty)$ as the Legendre transforms of the log-moment generating functions, namely,
\begin{align}
    E_Q(\theta)&\triangleq\sup_{\lambda\in\mathbb{R}}\pp{\lambda\theta-\psi_Q(\lambda)},\\
    E_P(\theta)&\triangleq\sup_{\lambda\in\mathbb{R}}\pp{\lambda\theta-\psi_P(\lambda)},\label{eq:chernoff}
\end{align}
where $\psi_Q(\lambda)\triangleq\log\bE_Q[\exp(\lambda\calL_{\s{I}}(A,B))]$ and $\psi_P(\lambda)\triangleq\log\bE_P[\exp(\lambda\calL_{\s{I}}(A,B))]$. Closed-form expressions for $\psi_Q$ and $\psi_P$ in our setting, can be found in Lemma~\ref{lem:simplecalc}. Then, applying standard Chernoff's bound on these tails yield that if the threshold $\tau_{\s{count}}$ is such that $\tau_{\s{count}}\in(-d_{\s{KL}}(Q_{XY}||P_{XY}),d_{\s{KL}}(P_{XY}||Q_{XY}))$, then
\begin{subequations}\label{eqn:chen}
\begin{align}
    \calQ_{\rho}&\leq \exp\pp{- E_Q(\tau_{\s{count}})},\label{eqn:chen1}\\
    \calP_{\rho}&\geq 1-\exp\pp{- E_P(\tau_{\s{count}})}.\label{eqn:chen2}
\end{align}
\end{subequations}
%Note that $\psi_P(\lambda) = \psi_Q(\lambda+1)$, and thus $E_P(\theta) = E_Q(\theta)-\theta$. In particular, $E_P$ and $E_Q$ are non-negative convex functions. Moreover, since $\psi'_Q(0) = -d_{\s{KL}}(Q_{XY}||P_{XY})$ and $\psi'_Q(1) = d_{\s{KL}}(P_{XY}||Q_{XY})$, we have $E_Q(-d_{\s{KL}}(Q_{XY}||P_{XY})) = E_P(d_{\s{KL}}(P_{XY}||Q_{XY}))=0$ and hence $E_Q(d_{\s{KL}}(P_{XY}||Q_{XY})) = d_{\s{KL}}(P_{XY}||Q_{XY})$ and $E_P(-d_{\s{KL}}(Q_{XY}||P_{XY}))=d_{\s{KL}}(Q_{XY}||P_{XY})$. Finally, it is well known that $\lambda\theta-\psi_Q(\lambda)$ is concave and has derivative at zero given by $\theta+d_{\s{KL}}(Q_{XY}||P_{XY})$. This implies that the maximizer $\lambda^\star$ to the concave optimization $E_Q(\theta)$ can be taken to be non-negative if $\theta\geq - d_{\s{KL}}(Q_{XY}||P_{XY})$. The same is true for $E_P(\theta)$ if  $\theta\leq d_{\s{KL}}(P_{XY}||Q_{XY})$.
We are now in a position to state our main result. %We mention here that our result holds for any natural $d\geq1$, while in \cite{HuleihelElimelech} it is assumed that $d\geq d_0$, for some fixed $d_0\in\mathbb{N}$ (most notably, excluding the interesting $d=1$ case). We have the following result.
\begin{theorem}[Count test strong detection]\label{thm:upper}
Consider the detection problem in \eqref{eqn:decproblem}, and the count test in \eqref{eqn:testcount}. Suppose there is a $\tau_{\s{count}}\in(-d_{\s{KL}}(Q_{XY}||P_{XY}),d_{\s{KL}}(P_{XY}||Q_{XY}))$ with
\begin{subequations}\label{eqn:Countcond}
	\begin{align}
    E_Q(\tau_{\s{count}}) &= \omega\p{\log n},\label{eqn:Countcond1}\\
    E_P(\tau_{\s{count}}) &= \omega(n^{-1}).\label{eqn:Countcond2}
\end{align}
\end{subequations}
Then, $\s{R}(\phi_{\s{count}})\to0$, as $n\to\infty$. In particular, for \eqref{eqn:Countcond} to hold it suffices that $\rho^2 =1-o(n^{-4})$.
\end{theorem}
Recall that in Theorem~\ref{th:lowerStrong} we proved that strong detection is statistically impossible for any $\rho^2<1$; therefore, for strong detection to be possible $\rho^2$ must converge to unity. Theorem~\ref{thm:upper} shows that indeed the count test that achieves strong detection if $\rho^2$ approaches unity sufficiently fast. We would like to mention here that there is still a statistical gap between the lower and upper bounds in Theorem~\ref{th:lowerStrong} and Theorem~\ref{thm:upper}, respectively. While according to Theorem~\ref{th:lowerStrong} for strong detection to be possible, $\rho^2$ must converge to unity, as $n\to\infty$, it is not clear yet what is the sufficient rate of convergence. We suspect that Theorem~\ref{thm:upper} is not optimal, and that, slower rates are possible. An intriguing question is to analyze the optimal Neyman-Pearson test, which seems quite challenging similarly to other related problems where a latent combinatorial structure is planted. Finally, we would like to mention here that the strong detection upper bounds in \cite{nazer2022detecting,HuleihelElimelech} exclude the case where $d=1$, and as so Theorem~\ref{thm:upper} is novel.

Next, we consider weak detection. By definition, any test that achieves strong detection achieves weak detection automatically. As so, the count test above achieves weak detection under the same conditions stated in Theorem~\ref{thm:upper}. In fact, it is rather straightforward to show that weak detection is possible using the count test even if, for example, the condition $\rho^2=1-o(n^{-4})$ is replaced by the weaker condition $\rho^2 =1-\s{C}\cdot n^{-4}$, for some constant $\s{C}>0$ (or, more generally, the $\omega$-asymptotic in \eqref{eqn:Countcond} is replaced by some constant $\s{C}>0$). Nonetheless, for weak detection, it turns out that the following rather simpler test exhibits better performance guarantees. Let $\theta\in\mathbb{R}_+$, and define the test,
\begin{align}
    \phi_{\s{comp}}(\s{X},\s{Y})\triangleq\Ind\ppp{\abs{\sum_{i=1}^n(X_i-Y_i)}\leq\theta},\label{eqn:testcomp}
\end{align}
if $\rho\in(0,1]$, and we flip the direction of the inequality in \eqref{eqn:testcomp} if $\rho\in[-1,0)$. This test simply compares the sum of entries of $\s{X}$ and $\s{Y}$. Let $\s{G}\sim\calN(0,1)$ and $\s{G}'\sim\calN(0,1-|\rho|)$ We define the threshold $\theta$ as the value for which 
\begin{align}
    &d_{\s{TV}}(\calN(0,1),\calN(0,1-|\rho|))= \pr\p{|\s{G}|\geq\frac{\theta}{\sqrt{2n}}}-\pr\p{|\s{G}'|\geq \frac{\theta}{\sqrt{2n}}}.\label{eq:thetavalues}
\end{align}
Such a value exists by the definition of the total-variation distance for centered Gaussian random variables (see, e.g., \cite[pg. 10]{devroye2023total}). The intuition behind this choice of $\theta$ is that, as it turns out, the right-hand-side (r.h.s.) of \eqref{eq:thetavalues} is exactly the reward (i.e., $1-\s{R}(\phi)$) associated with the comparison test. Thus, since the total-variation distance at the left-hand-side of \eqref{eq:thetavalues} is clearly positive for any $\rho^2=\Omega(1)$, we obtain that the risk is strictly bounded by unity. We then have the following result.
\begin{theorem}[Comparison test weak detection]\label{thm:upper_comp}
Consider the detection problem in \eqref{eqn:decproblem}, and the comparison test in \eqref{eqn:testcomp}, with $\theta$ given by \eqref{eq:thetavalues}. If $\rho^2 = \Omega(1)$, then $\lim_{n\to\infty}\s{R}(\phi_{\s{comp}})<1$.
\end{theorem}
\begin{comment}
    \begin{proof}[Proof of Theorem~\ref{thm:upper_comp}]
    We analyze the case where $\rho\in(0,1]$, with the understanding that the case where $\rho\in[-1,0)$ is analyzed in the same way. Let $G_1\triangleq\sum_{i=1}^nX_i$ and $G_2\triangleq\sum_{i=1}^nY_i$. Then, under $\calH_0$, we clearly have $G_1-G_2\sim\calN(0,2n)$, while under $\calH_1$, we have $G_1-G_2\sim\calN(0,2n(1-\rho))$. Therefore, 
    \begin{align}
        1-\s{R}(\phi_{\s{comp}}) &= \pr_{\calH_0}(|G_1-G_2|\geq\theta)\nonumber\\
        &\quad-\pr_{\calH_1}(|G_1-G_2|\geq\theta)\\
        & = \pr(|\calN(0,2n)|\geq\theta)\nonumber\\
        &\quad-\pr(|\calN(0,2n(1-\rho))|\geq\theta)\\
        & = d_{\s{TV}}(\calN(0,1),\calN(0,1-\rho))\\
        &= \Omega(1),
    \end{align}
    where the third equality holds by the definition of $\theta$, and the last equality is because $\rho = \Omega(1)$. 
\end{proof}
\end{comment}
Comparing Theorems~\ref{thm:upper} and \ref{thm:upper_comp} we see that for weak detection we can construct a test for which it is suffice that $\rho^2$ is of order constant (i.e., bounded away from zero), while for strong detection our test requires the correlation to converge to unity sufficiently fast. We would like to emphasize that Theorem~\ref{thm:upper_comp} is novel; in particular, the weak detection results in \cite{nazer2022detecting,HuleihelElimelech} again exclude the $d=1$ case (and in fact a range of values of $d$, as discussed in more detail in Section~\ref{sec:extension}).  
We provide the proofs of Theorems~\ref{thm:upper} and \ref{thm:upper_comp} in Section~\ref{swubsec:proofUpperCount}, where in fact, we prove a more general result for a high-dimensional model which allows for partial correlations, as described in Section~\ref{sec:extension}. We summarize the above results in Table~\ref{tab:exact}, next to previously known bounds in the literature. Finally, we conclude this subsection by showing a numerical evaluation of our count test. Specifically, in Fig.~\ref{fig:stupid}, we present the empirical risk, averaged over $10^3$ Monte-Carlo runs, associated with the count test in \eqref{eqn:testcount} with $\tau_{\s{count}}=0$, as a function of $\rho$, for $n=100$. As predicted by our theoretical results, it can be seen that if the correlation is sufficiently close to unity, then the associated risk decreases. 
\begin{table}[t!]
\begin{center}
\renewcommand{\arraystretch}{2}
\begin{tabular}{ |p{3.5cm}||p{1.8cm}|p{2.1cm}|p{1.8cm}| p{2.1cm}|}
 \hline 
   & \multicolumn{2}{|c|}{\textbf{Weak Detection}} &  \multicolumn{2}{|c|}{\textbf{Strong Detection}}  \\
 \hline
  \textbf{Current/previous work} & \textbf{Possible} & \textbf{Impossible} &\textbf{Possible}& \textbf{Impossible}\\
 \hline 
 Our results  & $\underset{\text{(Theorem~\ref{thm:upper_comp})}}{\Omega(1) }$    & $\underset{\text{(Theorem~\ref{th:lowerWeak})}}{o(1)}$ &     $\underset{\text{(Theorem~\ref{thm:upper})}}{1-o(n^{-4}) }$    & $\underset{\text{(Theorem~\ref{th:lowerStrong})}}{1-\Omega(1)} $\\
  \hline
 Previous work  & --    & $o(1)$ &     $- $    & $\frac{1}{2}$\\ 
  \hline
\end{tabular}
\caption{A summary of our bounds on $\rho^2$, for weak and strong detection in the one-dimensional case, compared with the results of \cite{HuleihelElimelech}.}
\label{tab:exact}
\end{center}
\end{table}

\begin{figure}[ht]
\centering
\begin{overpic}[scale=1]
    {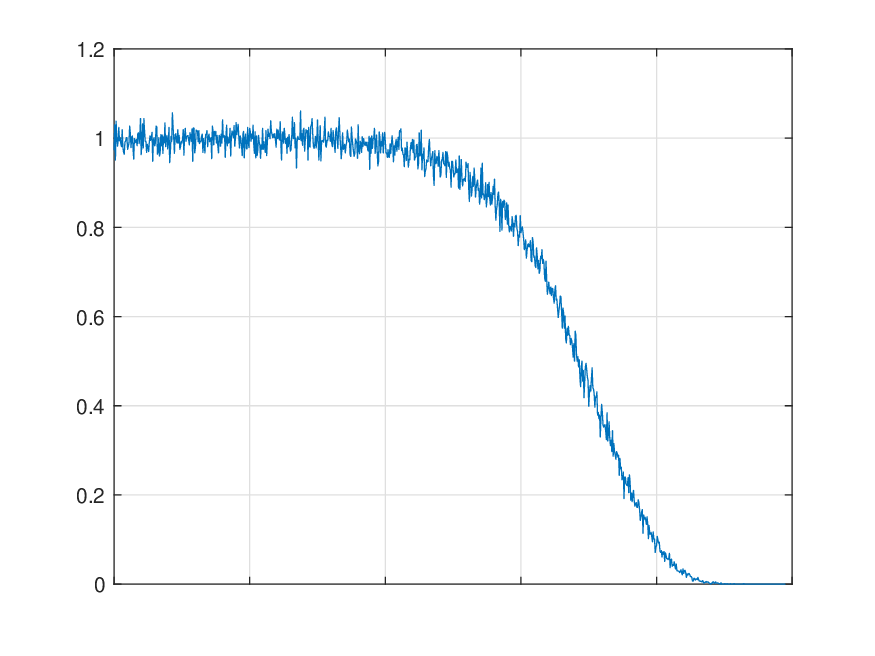}
    
    \put(2.5,37.8){\begin{turn}{90}\large{$\s{R}(\phi_{\s{count}})$}\end{turn}}
    \put(50.5,2.7){\Large{$\rho^2$}}
\end{overpic}
\caption{
The risk of the count test $\phi_{\s{count}}$ as a function of $\rho$, for $d=1$, $n=100$ and values of $\rho^2$ spanning between $1-n^{-2}$ and $1-n^{-5}$. The results are compatible with Theorem~\ref{th:lowerStrong}, showing that the risk indeed vanishes when $\rho^2=1-o(n^{-4}$).}
\label{fig:stupid}
\end{figure}

\section{Lower Bound via Polynomial Decomposition}\label{sec:LBPOL}
This section is devoted to the proofs of Theorems~\ref{th:lowerStrong}--\ref{th:lowerWeak}. \textcolor{black}{ A crucial part of our proof involves the expansion of the likelihood function into its orthogonal Hermite polynomials components, which compose an orthonormal basis of $L^2$ w.r.t. the Gaussian measure. We shall therefore begin with a brief introduction of Hermite polynomials in Hilbert spaces. For the benefit of readability, throughout this entire section, with some abuse of notation, we use lower-case Greek letters (e.g., $\alpha$), for integer vectors of finite lengths.%, where the length of vectors is mentioned of understood by context.
}

\subsection{Hermite polynomials and Hilbert spaces}\label{sec:Hilbert}
Let $\s{Y}\sim \calN(0,\mathbf{I})$ be an $n$-dimensional random variable distributed according to, $\mu$, an $n$-dimensional Gaussian distribution. Consider the space $L^2(\mu)$,  of real-valued random variables $\R^n\to\R$ with finite second moment (w.r.t. the measure to the measure $\mu$), which is a Hilbert space when equipped with the inner product,
\begin{align}
    \innerP{\phi,\psi}_{\mu}\eqdef \E_{\mu}\pp{\phi(\s{Y})\cdot \psi(\s{Y})},\label{eq:innerPs}
\end{align}
where $\phi$ and $\psi$ are two measurable functions from $\R^n$ to $\R$.  Recall the univariate Hermite polynomials, which are a sequence of polynomials $(h_k(x))_{k\geq0}$, with $\s{deg}(h_k)=k$, defined by the equation,
\begin{align}
    h_k(x)\triangleq(-1)^k e^{x^2/2}\frac{d^k}{dx^k}e^{-x^2/2}.
\end{align} It is well-known  (see, for example \cite{Magnus}[Chapter 5.6]) that Hermite polynomials are orthonormal w.r.t. the standard one-dimensional Gaussian measure. Namely, 
\begin{align}\label{eq:herOrthhh}
    \bE_{Y\sim N(0,1)}\pp{h_k(Y)h_\ell(Y)} = \delta[k-\ell].
\end{align}
The multivariate Hermite polynomials in $n$ variables are indexed by $\theta\in\mathbb{N}^{n}$, and are merely products of the univariate Hermite polynomials, i.e., $H_\theta(x) = \prod_{i=1}^n h_{\theta_{i}}(x_{i})$. We note that the degree of $H_\theta$ is exactly $|\theta|\eqdef \sum_i\theta_i$, which is exactly the $L_1$ norm of the vector $\theta$. Furthermore, the multivariate Hermite polynomials $(H_{\theta})_{\theta \in \N^n}$ are orthonormal w.r.t. an independent product of standard Gaussian measures, namely,
\begin{align}
    \bE_{\s{Y}\sim N(\mathbf{0},\mathbf{I})}\pp{H_\alpha(\s{Y})H_\gamma(\s{Y}) }&=\bE_{\s{Y}\sim N(\mathbf{0},\mathbf{I})}\pp{\prod_{i=1}^n h_{\alpha_i}(Y_i) h_{\gamma_i}(Y_i)}\\
    &\overset{(a)}{=}\prod_{i=1}^n \bE_{Y_i\sim N(\mathbf{0},1)}\pp{h_{\alpha_i}(Y_i) h_{\gamma_i}(Y_i)}\\&
    \overset{(b)}{=}\prod_{i=1}^n \delta[\alpha_i-\gamma_i] =\delta[\alpha-\gamma],\label{eq:HerOrth}
\end{align}
where $(a)$ follows from the independence of the random variables $\ppp{Y_i}_i$, and $(b)$ follows from \eqref{eq:innerPs}. 

It is well-known that the Hermitian polynomials form a complete orthonormal basis in the Hilbert space $L^2(\mu)$ w.r.t. the Gaussian measure (see, e.g., \cite[Proposition 1.10]{urbina2019preliminary}), and therefore any random variable of the form $\psi(\s{Y})$ can be expanded as, 
\begin{align}
    \psi(\s{Y})=\sum_{\theta \in \N^{n}}\innerP{H_{\theta}(\s{Y}),\psi(\s{Y})}_{\mu}H_{\theta}(\s{Y}),\label{eq:conver}
\end{align}
where the equality is in the sense that the sum at the r.h.f.s of \eqref{eq:conver} converges in norm to $\psi(\s{Y})$. Finally, Parseval's identity implies that,
\begin{align}
\E_{\mu}\pp{\psi(\s{Y})^2}&=\norm{\psi(\s{Y})}^2_{\mu}\\
    &=\sum_{m=0}^\infty \sum_{\substack{\theta\in \N^n\\ |\theta|=m }}\abs{\innerP{H_{\theta}(\s{Y}),\psi(\s{Y})}_{\mu}}^2. \label{eq:parseval1}
\end{align}

\subsection{Proof outline} \label{sec:outline}
Before delving into the proofs, let us provide a brief outline of the main steps. Roughly speaking, we derive our lower bounds using a non-standard analysis of the the second moment of likelihood function which involves the decomposition of the likelihood function into its orthogonal components, which we presented in the previous subsection. The main steps are:
\begin{enumerate}
    \item \textbf{Bounding the risk using the likelihood's second moment:} Using a series of fairly standard arguments, we bound the optimal risk function as
    \begin{align}\label{eq:first}
        \s{R}^\star\geq 1-\frac{1}{2}\sqrt{\E_{\calH_0}\pp{\calL(\s{X},\s{Y})^2}-1}.
    \end{align}
    \item \textbf{Hermite polynomial decomposition:} Using Parseval's identity we present the likelihood's second moment (with respect to $\calH_0$) using the orthogonal projection coefficients,
    \begin{align}\label{eq:Parse}
    \E_{\calH_0}\pp{\calL(\s{X},\s{Y})^2}=\sum_{(\alpha,\beta)\in \N^{2n}}\abs{\innerP{H_{\alpha,\beta}(\s{X},\s{Y}),\calL(\s{X},\s{Y})}_{\calH_0}}^2,
\end{align}
    where $H_{\alpha,\beta}$ is the Hermite polynomial with $2n$ variables which corresponds to the pair $(\alpha,\beta)\in \N^{2n}$.
  \item \textbf{Computing the orthogonal projection:}  In Lemma~\ref{lem:OrthCeoff}, we show that the coefficients of the orthogonal projection coefficients  can be presented as,
\begin{align}
    \innerP{H_{\alpha,\beta}(\s{X},\s{Y}),\calL(\s{X},\s{Y})}_{\calH_0}=\E_{\sigma}\pp{\E_{\calH_1|\sigma}\pp{H_{\alpha,\beta}(\s{X},\s{Y})}}.
\end{align}
Then, we carefully analyse the r.h.s. of the above equation using some well-known properties of Hermite polynomials and give an expression which depends only on $\alpha$ and $\beta$. 
\item \textbf{Integer partitions:} We plug in the result in the previous step into \eqref{eq:Parse}, and show that the likelihood's second moment is uniformly bounded by the generating function of integer partitions (see, Definition~\ref{def:ppp}). Finally, we combine this bound and \eqref{eq:first} in order to get a threshold for the impossibility of detection.
\end{enumerate}
\subsection{The proof of the lower bound}

The proof strategy for our lower bound begins with a series of standard arguments frequently used in  proofs of impossibility of detection in hypothesis testing problems involving a planted combinatorial structure. We start by recalling the well-known characterization of the optimal Bayes risk, using the total-variation distance \cite[Theorem 2.2]{tsybakov2004introduction},
\begin{align}
    \s{R}^\star = 1-d_{\s{TV}}(\P_{\calH_0},\P_{\calH_1}),\label{eq:Rstar}
\end{align}
where $d_{\s{TV}}$ denotes the total variation distance between two probability measures. In the case where $\P_{\calH_1}$ is absolutely continuous w.r.t. $\P_{\calH_0}$ the total variation distance is given by
\begin{align}
    d_{\s{TV}}(\P_{\calH_0},\P_{\calH_1})=\frac{1}{2} \intop \abs{\calL(\s{X},\s{Y})-1}\mathrm{d} \P_{\calH_0},\label{eq:dTV}
\end{align}
where $\calL(\s{X},\s{Y})$ is the  Radon-Nikodym derivative of $\P_{\calH_1}$ w.r.t. $\P_{\calH_0}$, also known as the\textit{ likelihood ratio}. In our case $\calL(\s{X},\s{Y})$ is given by
\begin{align}
    \calL(\s{X},\s{Y})\eqdef 
\frac{\mathrm{d}\P_{\calH_1}}{\mathrm{d}\P_{\calH_0}}=\E_{\sigma}\pp{\frac{f_{\calH_1|\sigma}(\s{X},\s{Y})}{f_{\calH_0}(\s{X},\s{Y})}},\label{eq:radonNi}
\end{align}
where $f_{\calH_0}(\s{X},\s{Y})$
denotes the density function of $(\s{X},\s{Y})$ and  $f_{\calH_1|\sigma}(\s{X},\s{Y})$
denotes the density function of $(\s{X},\s{Y})$ conditioned that hidden permutation is $\sigma\sim \s{Unif}(\S_n)$. The second step involves a standard use of Cauchy-Schwartz inequality together with 
\eqref{eq:Rstar} and $\eqref{eq:dTV}$ in order to obtain: 
\begin{align}
    \s{R}^\star&\geq 1-\frac{1}{2}\sqrt{\E_{\calH_0}\pp{\calL(\s{X},\s{Y})^2}-1}\\
    & = 1-\frac{1}{2}\sqrt{\chi^2(\P_{\calH_0},\P_{\calH_1})}\label{eq:CauchyS},
\end{align}
where $\chi^2$ is the $\chi^2$-divergence/distance \cite{tsybakov2004introduction}. The above equation shows that if $\E_{\calH_0}[\calL^2]$ approaches to $1$, then weak and strong detection are impossible.   

At this point, we deviate from the typical straight-forward analysis of the second moment of the likelihood ratio as in \cite{HuleihelElimelech, paslev2023testing,nazer2022detecting}, which involves a non-trivial analysis of the cycles structure of random permutations. Not only that this analysis is highly complicated, as we show, it also results in a loose lower bound. Instead, we suggest a simple, but powerful alternative technique for bounding the second moment of the likelihood ratio, using a decomposition/expansion of the likelihood function $\calL(\s{X},\s{Y})$ w.r.t. the orthonormal basis of Hermite polynomials.

The key idea of our technique is to decompose the likelihood ratio into its orthogonal components using the basis formed by the (infinite) family of Hermite polynomials, and then use Parseval's identity. As mentioned in Section~\ref{sec:Hilbert}, the Hermitian polynomials form a complete orthonormal system in $L^2(\calH_0)$, and by \eqref{eq:parseval1},  $\calL(\s{X},\s{Y})$ can be expanded as, 
\begin{align}
    \calL(\s{X},\s{Y})=\sum_{\alpha,\beta\in \N^{n}}\innerP{H_{\alpha,\beta}(\s{X},\s{Y}),\calL(\s{X},\s{Y})}_{\calH_0}H_{\alpha,\beta}(\s{X},\s{Y}),
\end{align}
where $H_{\alpha,\beta}$ denotes the Hermite polynomial with $2n$ variables defined by $(\alpha,\beta)\in \R^{2n}$. The above sum converges in norm to $\calL$, and by Parseval's identity,
\begin{align}
\E_{\calH_0}\pp{\calL(\s{X},\s{Y})^2}&=\norm{\calL(\s{X},\s{Y})}^2_{\calH_0}\\
    &=\sum_{m=0}^\infty \sum_{\substack{(\alpha,\beta) \in \N^{2n}\\ |\alpha|+|\beta|=m }}\innerP{H_{\alpha,\beta}(\s{X},\s{Y}),\calL(\s{X},\s{Y})}_{\calH_0}^2. \label{eq:parseval}
\end{align}
In the following lemma, we find an expression for the projection coefficients in \eqref{eq:parseval}. \begin{lemma}\label{lem:OrthCeoff}
    For any $(\alpha,\beta) \in \N^{2n}$, 
    \begin{align}
        \innerP{H_{\alpha,\beta}(\s{X},\s{Y}),\calL(\s{X},\s{Y})}_{\calH_0}=\rho^{|\alpha|}\cdot \P[\sigma(\beta)=\alpha],
    \end{align}
    where $\sigma(\alpha)\in \N^d$ denotes the vector obtained by permuting the coordinates of $\alpha$ according to the uniformly distributed random permutation $\sigma$.
\end{lemma}

\begin{proof}
    By the definition of the Radon-Nikodym derivative we have,
    \begin{align}
        \innerP{H_{\alpha,\beta}(\s{X},\s{Y}),\calL(\s{X},\s{Y})}_{\calH_0}&=\E_{\calH_0}\pp{H_{\alpha,\beta}(\s{X},\s{Y})\cdot\calL(\s{X},\s{Y})}\\
        & =\E_{\calH_0}\pp{H_{\alpha,\beta}(\s{X},\s{Y})\frac{\mathrm{d}\P_{\calH_1}}{\mathrm{d}\P_{\calH_0}}}\\
        &=\E_{\calH_1}\pp{H_{\alpha,\beta}(\s{X},\s{Y})}\\
        &=\E_{\sigma}\pp{\E_{\calH_1|\sigma}\pp{H_{\alpha,\beta}(\s{X},\s{Y})}}\\
        &=\E_{\sigma}\pp{\E_{\calH_1|\sigma}\pp{\prod_{i=1}^n h_{\alpha_i}(\s{X}_i) 
        \cdot h_{\beta_{\sigma(i)}}(\s{Y}_{\sigma(i)})}}\\
        &=\E_{\sigma}\pp{\prod_{i=1}^n \E_{\calH_1|\sigma}\pp{h_{\alpha_i}(\s{X}_i) 
        \cdot h_{\beta_{\sigma(i)}}(\s{Y}_{\sigma(i)})}}, \label{eq:prod1}
    \end{align}
    where the last equality follows from the independence of the pairs $(\s{X}_i,\s{Y}_{\sigma(i)})_{i=1}^n$ conditioned on the hidden permutation $\sigma$. Given $\sigma$, the pair of random variables $(\s{X}_i,\s{Y}_{\sigma(i)})$ is distributed as the two-dimensional Gaussian distribution $\calN(\mathbf{0},\Sigma_\rho)$. Let $X$ and $Z$ be two independent standard normal random variables. A simple calculation reveals that the pair $(X,\rho X + \sqrt{1-\rho^2}Z)$ is also distributed as $\calN(\mathbf{0},\Sigma_\rho)$, and therefore for all $i$,
    \begin{align}
    \E_{\calH_1|\sigma}\pp{h_{\alpha_i}(\s{X}_i) 
        \cdot h_{\beta_{\sigma(i)}}(\s{Y}_{\sigma(i)})}
        &=\E_{X\indep Z} \pp{ h_{\alpha_i}(X) \cdot h_{\beta_{\sigma(i)}}(\rho X + \sqrt{1-\rho^2}Z)}\\
        &=\E\pp{\left.h_{\alpha_i}(X)\cdot \E\pp{h_{\beta_{\sigma(i)}}(\rho X + \sqrt{1-\rho^2}Z)\right|X}}, \label{eq:prod2}
    \end{align}
       where the last equality follows from the law of total expectation. When conditioned on $X$, the random variable $\rho X + \sqrt{1-\rho^2}Z$ satisfies 
       \begin{equation}\label{eq:conditionedonx}
           (\rho X + \sqrt{1-\rho^2}Z) ~\big|~ X \sim \calN(\rho X,1-\rho^2).
       \end{equation} Thus, our goal is now to find the expectation $\E_{Y\sim \calN(\rho x,1-\rho^2)}[h_k(Y)]$, for any $k\in\N$ and $|\rho|< 1$.    Let $\bar{h}_k$ denote the $k$-th order physicist's Hermite polynomials. It is known that \cite[equation (3)]{FrancisBach},
\begin{align}\label{eq:HerIden}
    \int_{\mathbb{R}} \bar{h}_k(y)e^{-\frac{(y-\rho x)^2}{1-\rho^2}}\mathrm{d}y= \sqrt{\pi} \rho^k \sqrt{1-\rho^2} \cdot \bar{h}_k (x),
\end{align}
for $|\rho|<1$, $k\geq0$, and $y\in\mathbb{R}$. Also, it is known that \cite[Chapter 5.6]{Magnus} 
\begin{equation}
    \label{eq:HemitePhys}h_k(x) = 2^{-\frac{k}{2}}\bar{h}_k\p{\frac{x}{\sqrt{2}}}.
\end{equation} Thus,
\begin{align}
    \bE_{Y\sim N(\rho x,1-\rho^2)}\pp{h_k(Y)} &= \frac{1}{\sqrt{2\pi(1-\rho^2)}}\int_{\mathbb{R}}h_k(y)e^{-\frac{(y-\rho x)^2}{2(1-\rho^2)}}\mathrm{d}y\\
    & =\frac{2^{-\frac{k}{2}}}{\sqrt{2\pi(1-\rho^2)}}\int_{\mathbb{R}}\bar{h}_k\p{\frac{y}{\sqrt{2}}}e^{-\frac{(y-\rho x)^2}{2(1-\rho^2)}}\mathrm{d}y\\
    & = \frac{2^{-\frac{k}{2}}}{\sqrt{\pi(1-\rho^2)}}\int_{\mathbb{R}}\bar{h}_k\p{y}e^{-\frac{(y-\rho x/\sqrt{2})^2}{(1-\rho^2)}}\mathrm{d}y\\
    & \overset{(a)}{=} \frac{2^{-\frac{k}{2}}}{\sqrt{\pi(1-\rho^2)}}\sqrt{\pi} \rho^k \sqrt{1-\rho^2} \bar{h}_k(x/\sqrt{2})\label{eq:50}\\
    & = 2^{-\frac{k}{2}}\rho^k\bar{h}_k(x/\sqrt{2})\\
    & = \rho^kh_k(x).  \label{eq:prod3}
\end{align}
where $(a)$ follows by using \eqref{eq:HerIden} with $\Tilde{x}=x/\sqrt{2}$. Combining \eqref{eq:prod2}, \eqref{eq:conditionedonx},  and \eqref{eq:prod3} we obtain, 
    \begin{align}
        \E_{\calH_1|\sigma}\pp{h_{\alpha_i}(\s{X}_i) 
        \cdot h_{\beta_{\sigma(i)}}(\s{Y}_{\sigma(i)})}
        &=\E\pp{h_{\alpha_i}(X)\cdot \E\pp{h_{\beta_{\sigma(i)}}(\rho X + \sqrt{1-\rho^2}Z)~\big|~X}}\\
        &=\E\pp{h_{\alpha_i}(X)\cdot  \bE_{Y\sim N(\rho X,1-\rho^2)}\pp{h_{\beta_{\sigma(i)}}(Y)}}\\
        &=\E\pp{h_{\alpha_i}(X)\cdot \rho^{\beta_{\sigma(i)}}\cdot h_{\beta_{\sigma(i)}}(X)}\\
        &=\rho^{\beta_{\sigma(i)}} \cdot \delta[\alpha_i-\beta_{\sigma(i)}],
    \end{align}
where the last equality follows from the orthogonality property of the Hermite polynomials. We now conclude,
    \begin{align}
         \innerP{H_{\alpha,\beta}(\s{X},\s{Y}),\calL(\s{X},\s{Y})}_{\calH_0}&=\E_{\sigma}\pp{\prod_{i=1}^n \E_{\calH_1|\sigma}\pp{h_{\alpha_i}(\s{X}_i) 
        \cdot h_{\beta_{\sigma(i)}}(\s{Y}_{\sigma(i)})}}\label{eq:40}\\
        &=\E_{\sigma}\pp{\prod_{i=1}^n \rho^{\beta_{\sigma(i)}} \cdot \delta[\alpha_i-\beta_{\sigma(i)}]}\\
        &=\E_{\sigma}\pp{\rho^{\sum_{i=1}^n \beta_{\sigma(i)}} \Ind_{\sigma(\beta)=\alpha}}\\
        &=\rho^{|\beta|} \P[\sigma(\beta)=\alpha].
    \end{align}
    Note that the if $\P[\sigma(\beta)=\alpha]\neq0$ then it must be that $|\alpha|=|\beta|$ (since the $L_1$ norm is invariant under reordering of the coordinates). Thus, we have
    \begin{equation}\label{eq:reored}
         \innerP{H_{\alpha,\beta}(\s{X},\s{Y}),\calL(\s{X},\s{Y})}_{\calH_0}\rho^{|\beta|}=\P[\sigma(\beta)=\alpha]=\rho^{|\alpha|} \P[\sigma(\beta)=\alpha].
    \end{equation}
\end{proof}

We now turn to analyse $\P[\sigma(\beta)=\alpha]$ for arbitrary vectors $\alpha,\beta \in \N^n$. To that end, we consider the integer distribution function of a vector $\alpha\in \N^n$, $p_\alpha:\N\to \N$, where $p_\alpha(\ell)$ is defined as the number of coordinates of $\alpha$ that contain $\ell$. We also consider the equivalence relation over $\N^n$ defined by setting $\alpha \equiv \beta$ if and only if there is $\sigma\in \S_n$ such that $\sigma(\beta)=\alpha$. We observe that $\alpha\equiv \beta$ if and only if they have the same distribution functions, namely, $p_\alpha=p_\beta$. Let $[\alpha]$ denote the equivalence class of $\alpha$ w.r.t. this equivalence relation.\footnote{For mathematical correctness, we can think of $[\cdot]$ as a map between integer vectors and equivalence classes, which maps any $\alpha$ to the unique equivalence class for which it belongs.} For $\sigma\sim\s{Unif}(\S_n)$, symmetry implies that $\sigma(\alpha)\sim \s{Unif}([\alpha])$. We therefore have,
\begin{align}
    \P[\sigma(\beta)=\alpha]=\frac{1}{|[\beta]|}\Ind_{\alpha\equiv \beta}=\frac{1}{|[\alpha]|}\Ind_{\alpha\equiv \beta}.\label{eq:probaper}
\end{align}
We are now ready to prove our lower bounds. By \eqref{eq:parseval} and \eqref{eq:probaper},
\begin{align}
    \E_{\calH_0}\pp{\calL(\s{X},\s{Y})^2}
    &=\sum_{m=0}^\infty \sum_{\substack{(\alpha,\beta) \in \N^{2n}\\ |\alpha|+|\beta|=m }}\innerP{H_{\alpha,\beta}(\s{X},\s{Y}),\calL(\s{X},\s{Y})}_{\calH_0}^2\\
    &=\sum_{m=0}^\infty \sum_{\substack{(\alpha,\beta) \in \N^{2n}\\ |\alpha|+|\beta|=m }}\frac{\rho^{m}}{|[\alpha]|^2}\Ind_{\alpha\equiv \beta}\label{eq:manysummands}\\
     &\overset{(a)}{=}\sum_{m=0}^\infty \sum_{\substack{(\alpha,\beta) \in \N^{2n}\\ |\alpha|=|\beta|=m }}\frac{\rho^{2m}}{|[\alpha]|^2}\Ind_{\alpha\equiv \beta}\\
     &=\sum_{m=0}^\infty \sum_{\substack{\alpha\in \N^n\\ |\alpha|=m }}\sum_{\substack{\beta\in \N^n\\ \beta\equiv \alpha}}\frac{\rho^{2m}}{|[\alpha]|^2}\\
     &=\sum_{m=0}^\infty \sum_{\substack{\alpha\in \N^n\\ |\alpha|=m }}\frac{\rho^{2m}}{|[\alpha]|^2}\sum_{\beta\in [\alpha]}1\\
     &=\sum_{m=0}^\infty \sum_{\substack{\alpha\in \N^n\\ |\alpha|=m }}\frac{\rho^{2m}}{|[\alpha]|}\\
     &\overset{(b)}{=}\sum_{m=0}^{\infty}\sum_{\substack{[\alpha]\\
      \abs{\alpha}=m}}\sum_{\beta\in[\alpha ]}\frac{\rho^{2|\beta|}}{|[\beta]|}\\
      &\overset{(c)}{=}\sum_{m=0}^{\infty}\sum_{\substack{[\alpha]\\
      \abs{\alpha}=m}}\rho^{2m}\\&=\sum_{m=0}^{\infty}\abs{\ppp{[\alpha]: |\alpha|=m}}\cdot \rho^{2m}, \label{eq:parsevalDec}
\end{align}
where: 
\begin{itemize}
    \item[$(a)$]  follows since the $L_1$ norm is invariant under the equivalence relation $\equiv$ (since reordering of the coordinate do not affect the $L_1$ norm) and therefore the summand in \eqref{eq:manysummands} may be nonzero only if $|\alpha|=|\beta|$. Thus, the summation includes only pairs $(\alpha,\beta)$ such that $|\alpha|=|\beta|=m$ for some $m\in \N$.
    \item[$(b)$] follows since the set of equivalence classes is a partition of $\N^n$,
    \item[$(c)$] follows since $[\beta]=[\alpha]$ for $\alpha\equiv \beta$.
\end{itemize}

For a fixed $m\in \N$, consider the set of all equivalence classes $([\alpha])_{\alpha\in \N^n}$ with $\abs{\alpha}=m$ (which is well defined as the $L^1$ norm is preserved under equivalence relation  $\equiv$). As we mentioned, $\alpha\equiv \beta$ if and only if $p_\alpha=p_\beta$. We also note that for any $\alpha$, with corresponding distribution function $p_{\alpha}$, the sum $\sum_{k\in \N} p_\alpha(k)\leq n$ is the number of non-zero elements of $\alpha$, and the sum $\sum_{k\in \N} k\cdot p_\alpha(k)=|\alpha|$ is the degree of the polynomial $H_\alpha$.
Thus, equivalence classes of the form $([\alpha])_{\alpha\in \N^n}$, $|\alpha|=m$, are in  a one-to-one correspondence with integer distribution functions $p:\N\to \N$ which satisfy
\begin{align}
    \sum_{k\in \N} k\cdot p(k)=m\quad\text{and}\quad\sum_{k\in \N}p(k)\leq n. \label{eq:condDist}
\end{align} 
Such integer distributions represent integer partitions, which are defined as follows.
\begin{definition}\label{def:ppp}
    An integer partition of the number $m$ is an unordered  multiset of positive integers $\calp =\set{a_1,\dots,a_k}$ such that  $\sum_{i=1}^k a_i=m$. We denote the set of integer partitions of $m$ to exactly $n$ elements by $\s{Par}(m,n)$ and the set of integer partitions of $m$ to at most $n$ elements by $\s{Par}(m,\leq_n)$.
\end{definition}
Indeed, any distribution function which satisfy the conditions in \eqref{eq:condDist} correspond to a unique partition $\calp$ of $m$ to at most $n$ elements given by the multiset in which any integer $k\in\N$ has exactly $p(k)$ copies. Hence, we conclude that,
\begin{align}
    |\ppp{[\alpha]: |\alpha|=m}|
    &=\abs{\ppp{p:\N\to \N ~\big|~ \sum_{k\in \N} k\cdot p(k)=m, \sum_{k\in \N}p(k)\leq n}}\\
    &=\abs{\s{Par}(m,\leq_n)}.\label{eq:partitionID}
\end{align}

We are now in a position to prove our lower bounds. Combining \eqref{eq:parsevalDec} and \eqref{eq:partitionID}, we bound the second moment of the likelihood ratio as follows,
\begin{align}
    \E_{\calH_0}[\calL(\s{X},\s{Y})^2]&=\sum_{m=0}^\infty\abs{\ppp{[\alpha]:|\alpha|=m}}\cdot \rho^{2m}\\
    &=\sum_{m=0}^\infty \abs{\s{Par}(m,\leq_n) }\rho^{2m}.\label{eq:ExactCalcL}
\end{align}
Since $\abs{\s{Par}(m,\leq_n)}$ is a monotone non-decreasing function of $n$ we have
\begin{align}
    \E_{\calH_0}[\calL(\s{X},\s{Y})^2]\leq\sum_{m=0}^\infty \abs{\s{Par}(m,\leq_\infty) }\rho^{2m},\label{eq:powersumbound}
\end{align}
where $|\s{Par}(m,\leq_{\infty})|$ is the number of integer partitions of the number $m$ (with no bound on the number of parts). We comment that bounding $|\s{Par}(m,\leq_{n})|$ by $|\s{Par}(m,\leq_{\infty})|$ (and thereby loosing the decency with $n$) does not compromise our analysis. In Remark~\ref{rem:cannotimprove} we prove a lower bound on $\E_{\calH_0}[\calL(\s{X},\s{Y})^2]$ which shows that the thresholds obtained cannot be improved by refining the analysis of \eqref{eq:ExactCalcL}.

In the following, we will use the well-known Hardy-Ramanujan Formula (see, e.g., \cite[Chapter 1.3]{flajolet2009analytic}) for the number of integer partitions, which states that there exists some universal constant $c\in \R$ such that for all $m\in \N\setminus\ppp{0}$,
\begin{align}
    |\s{Par}(m,\leq_\infty)|\leq c \cdot\frac{1}{4\sqrt{3}m}\exp\p{\pi\sqrt{\frac{2m}{3}}}.\label{eq:Ramanujan1}
\end{align}
\begin{proof}[Proof of Theorem~\ref{th:lowerWeak}] 
    By Hardy-Ramanujan Formula in \eqref{eq:Ramanujan1}, $|\s{Par}(m,\leq_\infty)|$ is sub-exponential in $m$, and the infinite power series in \eqref{eq:powersumbound} converges to a finite number, for any $\rho^2<1$, which we denote by $g(\rho)$. Since the limit function of a power series is continuous within the region of convergence (which in our case is the interval $(-1,1)$), for any sequence $(\rho_k)_k$, such that $\rho_k\to 0$, we have $g(\rho_k)\to g(0)=1$. When combined with \eqref{eq:powersumbound} and \eqref{eq:CauchyS}, we conclude that if $\rho^2=o(1)$ we have,
\begin{align}
    \s{R}^\star&\geq 1-\frac{1}{2}\sqrt{\E_{\calH_0}[\calL(\s{X},\s{Y})^2]-1}\\
    &\geq 1-\frac{1}{2}\sqrt{g(\rho)-1}\\
    &=1-\frac{1}{2}\sqrt{1+o(1)-1}=1-o(1).
\end{align}
That is, weak detection is impossible if $\rho^2=o(1)$.
\end{proof}
For the impossibility of strong detection, we recall the following well-known fact regarding the total-variation distance.
\begin{lemma}\cite[Lemma 2.6]{tsybakov2004introduction})
    For any sequence of measures $(\mu_n,\nu_n)_{n\in \N}$ on a measurable space such that $\nu_n$ is absolutely continuous w.r.t. $\mu_n$, if 
\begin{align}
    \E_{\mu_n}\pp{\p{\frac{\mathrm{d}\nu_n}{\mathrm{d}\mu_n}}^2}=O(1),
\end{align}
then 
\begin{align}
    d_{\s{TV}}(\mu_n,\nu_n)=1-\Omega(1).
\end{align} \label{lem:Omega(1)}
\end{lemma}
Note that Lemma~\ref{lem:Omega(1)} follows from the inequality $1-d_{\s{TV}}(\mu_n,\nu_n)\geq\frac{1/2}{1+\chi^2(\mu_n,\nu_n)}$. We are now in a position to prove Theorem~\ref{th:lowerStrong}.
\begin{proof}[Proof of Theorem~\ref{th:lowerStrong}] 
    By \eqref{eq:powersumbound} and \eqref{eq:Ramanujan1}, if $\rho^2=1-\Omega(1)$, then there exists $\varepsilon\geq 0$ such that $\rho_k^2\leq 1-\varepsilon $, for all $k\in \N$, and  we have,
\begin{align}
    \E_{\calH_0}[\calL(\s{X},\s{Y})^2]&\leq\sum_{m=0}^\infty \abs{\s{Par}(m,\leq_\infty) }\rho_k^{2m}
    \\&\leq \sum_{m=0}^\infty \abs{\s{Par}(m,\leq_\infty)}(1-\varepsilon)^{2m}\\
    &\leq \sum_{m=0}^\infty c \cdot\frac{1}{4\sqrt{3}m}e^{\pi\sqrt{2m/3}}(1-\varepsilon)^{2m}=O(1),\label{eq:ramnusum}
\end{align}
where the last equality follows since the sum in \eqref{eq:ramnusum} converges (for example, by the ratio test). Using Lemma~\ref{lem:Omega(1)} and $\eqref{eq:Rstar}$ we conclude that 
\[\s{R}^\star =1-d_{\s{TV}}(\P_{\calH_0},\P_{\calH_1})=\Omega(1).\]
\end{proof}
\begin{remark}\label{rem:cannotimprove}
    We observe that in \eqref{eq:ExactCalcL} we have an exact expression for $\E_{\calH_0}[\calL(\s{X},\s{Y})^2]$,  which we may bound from bellow by
    \begin{align}
        \E_{\calH_0}\pp{\calL(\s{X},\s{Y})^2}\geq \sum_{m\in \N}\rho^{2m}=\frac{1}{1-\rho^2},
    \end{align}
    which is bounded whenever $\rho^2=1-\Omega(1)$, and approaches $1$ whenever $\rho=o(1)$. This shows that the thresholds for the impossibility of weak and strong detection stated in Theorems~\ref{th:lowerStrong}--\ref{th:lowerWeak} cannot be further improved by a tighter analysis of $\E_{\calH_0}[\calL(\s{X},\s{Y})^2]$. This implies that the only scenario in which our lower bound can be improved is if the Cauchy-Schwartz inequality used in \eqref{eq:CauchyS} is not asymptotically tight.
\end{remark}

\section{Partial Correlation in High Dimensions}\label{sec:extension}

We now turn to extend our techniques and results to a generalized version of the one-dimensional setting in Section~\ref{sec:1D}. Specifically, we consider the problem of detecting the correlation between two Gaussian random databases, where the database elements are $d$-dimensional Gaussian random vectors (rather than a random variable), and further consider the case in which the databases are only partially correlated. Let us formulate this mathematically.

Our generalized model is formulated again as a binary hypothesis testing problem. Under the null hypothesis $\calH_0$, the Gaussian databases $\s{X}\in\mathbb{R}^{n\times d}$ and $\s{Y}\in\mathbb{R}^{n\times d}$ are generated independently at random with $X_1,\ldots,X_n,Y_1,\ldots,Y_n\sim \calN(0_d,\mathbf{I}_d)$. Under the alternative hypothesis $\calH_1$, $k$ out of the $n$ elements from the first database $\s{X}$ are correlated with some other $k$ elements in a randomly and uniformly permuted version of $\s{Y}$. Specifically, fix a parameter $k\leq n$, and let $\calK$ be a set of $k$ indices chosen uniformly at random from $\{1,2,\dots,n\}$, and let $\sigma$ be a permutations of $[n]$, chosen uniformly at random as well. Then, for each $i\in\mathcal{K}$ and $1\leq j\leq d$, under $\calH_1$, the random variables $X_{ij}$ and $Y_{\sigma(i)j}$ are jointly Gaussian with correlation $\rho\neq 0$. To conclude, we deal with the following decision problem,
\begin{equation}
\begin{aligned}\label{eqn:decproblem2}
    &\calH_0: (X_1,Y_1),\ldots,(X_n,Y_n)\stackrel{\mathrm{i.i.d}}{\sim} \calN^{\otimes d}(\mathbf{0},\mathbf{I}_{2\times 2})\\
& \calH_1: \begin{cases}
\ppp{(X_i,Y_{\sigma(i)})}_{i\in \calK} \overset{\mathrm{i.i.d}}{\sim} \calN^{\otimes d}(\mathbf{0},\Sigma_\rho)\\
\ppp{(X_i,Y_{\sigma(i)})}_{i\notin \calK} \overset{\mathrm{i.i.d}}{\sim}\calN^{\otimes d}(\mathbf{0},\mathbf{I}_{2\times 2})\\
\ppp{(X_i,Y_{\sigma(i)})}_{i\notin \calK}\indep \ppp{(X_i,Y_{\sigma(i)})}_{i\in \calK},
\end{cases}
\end{aligned}
\end{equation}
where $\sigma\sim \s{Unif}(\mathbb{S}_n)$ independently of the random set $\calK\sim \s{Unif}\binom{[n]}{k}$, and we recall that,
\begin{align}
    \Sigma_\rho\triangleq\begin{bmatrix}
1 & \rho\\
\rho & 1
\end{bmatrix}.
\end{align}
The risk of a test $\phi:\R^{n\times d}\times\R^{n\times d}\to \{0,1\}$ is defined as \begin{align}
\s{R}(\phi)\triangleq \P_{\calH_0}[\phi(\s{X},\s{Y})=1]+\P_{\calH_1}[\phi(\s{X},\s{Y})=0],
\end{align}
and the Bayesian risk is 
\begin{align}
\s{R}^\star\triangleq\inf_{\phi:\R^{n\times d}\times\R^{n\times d}\to \{0,1\}}\s{R}(\phi).
\end{align}
As in the one-dimensional case, a   sequence $(\rho,n,k)=(\rho_\ell,n_\ell,k_\ell)_\ell$ is said to be:
\begin{enumerate}
    \item Admissible for strong detection if 
    $\lim_{\ell\to \infty}\s{R}^\star=0$.
    \item Admissible for weak detection if
    $\limsup_{\ell\to \infty} \s{R}^\star<1$.
\end{enumerate}

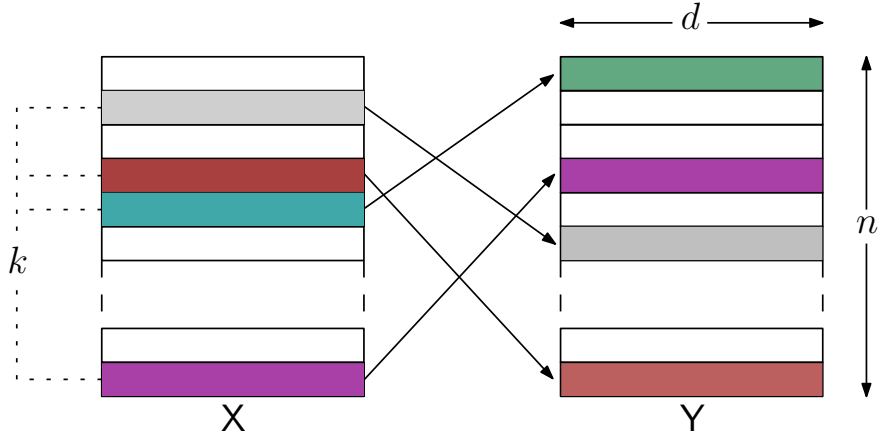
\begin{figure}[t]
\centering
\begin{overpic}[scale=0.09]
    {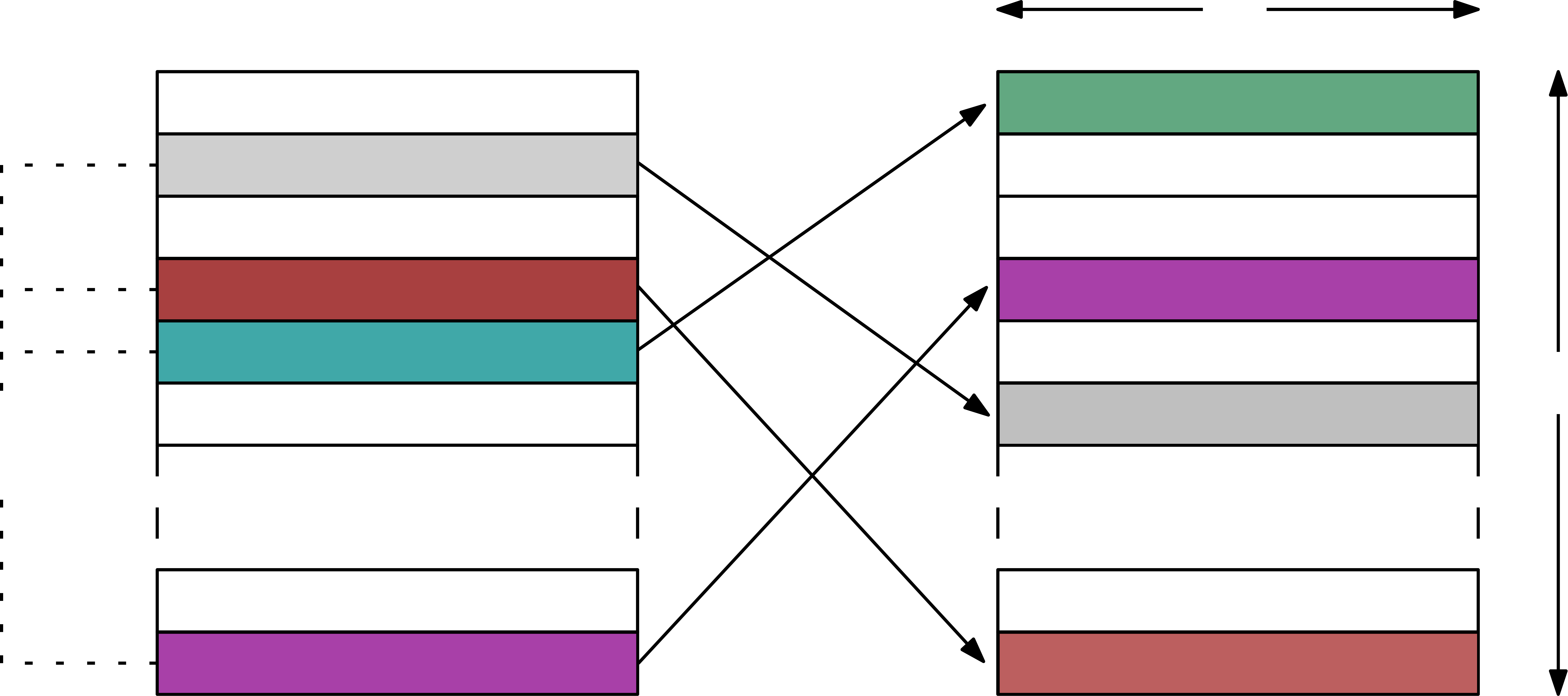}
      \put(-1,14.5){\large{$k$}} 
    \put(98.2,19.2){\large{$n$}} 
    \put(77.7,43){\large{$d$}} 
    \put(77.6,-4){\large{$\s{Y}$}} 
    \put(24,-4){\large{$\s{X}$}} 
    
\end{overpic}
\vspace{0.7cm}
\caption{
An illustration of the correlation structure under the alternative hypothesis $\calH_1$. Correlated elements are marked with a similar color, where the remaining mutually independent elements are colorless.
}
\label{fig:partial}
\end{figure}
\vspace{0.2cm}
\noindent\textbf{Lower bounds.} As in Section~\ref{sec:1D}, we start with our lower bounds. We have the following impossibility guarantees for strong and weak detection, respectively. 
\begin{theorem}[Impossibility of strong detection]
Consider the detection problem in \eqref{eqn:decproblem2}.
\label{th:lowerStrong2}
For any sequence $(\rho,n,k,d)=(\rho_\ell,n_\ell,k_\ell,d_\ell)_\ell$ such that $\rho^2<\frac{1}{d}$ and 
\begin{align}
    \p{\frac{k}{n}}^2\p{\prod_{i=1}^k\frac{1}{1-(d\rho^2)^i}-1}=O(1),\label{eqn:strongDec2}
\end{align} we have 
\begin{align}
   \s{R}^\star=\Omega(1).\label{eqn:kcondstrong}
\end{align}
Namely, strong detection is impossible.
\end{theorem}
\begin{theorem}[Impossibility of weak detection]
Consider the detection problem in \eqref{eqn:decproblem2}. For any sequence $(\rho,n,k,d)=(\rho_\ell,n_\ell,k_\ell,d_\ell)_\ell$ such that $\rho^2<\frac{1}{d}$ and 
\begin{align}
   \p{\frac{k}{n}}^2\p{\prod_{i=1}^k\frac{1}{1-(d\rho^2)^i}-1}=o(1),\label{eqn:kcondweak}
\end{align} we have 
\begin{align}
   \s{R}^\star=1-o(1).
\end{align}
Namely, weak detection is impossible.
\label{th:lowerWeak2}
\end{theorem}

Let us simplify the above results. Consider first the case where $k=n$, in which the left-hand-side of \eqref{eqn:kcondstrong}--\eqref{eqn:kcondweak} above becomes,
\begin{align}
    \prod_{i=1}^n\frac{1}{1-(d\rho^2)^i}-1.
\end{align}
We also observe that 
\begin{align}
    \prod_{i=1}^n\frac{1}{1-(d\rho^2)^i}\leq \prod_{i=1}^\infty\frac{1}{1-(d\rho^2)^i}\eqdef g(d\rho^2).\label{eq:thresanal}
\end{align}
The function $g$ at the right-hand-side of  \eqref{eq:thresanal} is the generating function of integer partitions (see, e.g., \cite[Chapter 1.3]{flajolet2009analytic}), and is given by,
\begin{align}
    g(x)=\sum_{m=0}^\infty \abs{\s{Par}(m,\leq_\infty)}x^m,
\end{align}
where $\s{Par}(m,\leq_\infty)$ is defined in Definition~\ref{def:ppp}. Note that for any $|x|<1$, $g(x)$ is a well-defined finite number. Hence, by \eqref{eq:thresanal}, if $d \rho^2\leq (1-\varepsilon)$, for some constant $\varepsilon>0$, then 
\begin{align}
    \prod_{i=1}^n\frac{1}{1-(d\rho^2)^i}-1\leq g(1-\varepsilon)=O(1).\label{eqn:fixeddStrongNew}
\end{align}
Similarly, if $d \rho^2=o(1)$ then
\begin{align}
    \prod_{i=1}^n\frac{1}{1-(d\rho^2)^i}-1\leq g(o(1))=1+o(1).\label{eqn:fixeddWeakNew}
\end{align}
Subsequently, \eqref{eqn:fixeddStrongNew} and \eqref{eqn:fixeddWeakNew} imply that strong and weak detection are impossible if $d \rho^2\leq (1-\varepsilon)$ and $d \rho^2=o(1)$ respectively. These conclusions coincide with the sharp thresholds in \cite{HuleihelElimelech}, in the regime where $d\to\infty$. Now, when $d$ is fixed, it was proved in \cite{HuleihelElimelech} that strong detection is impossible as long as $\rho^2<\rho^\star(d)$, for some function $\rho^\star(d)$. Accordingly, our result in \eqref{eqn:fixeddStrongNew} gives a new threshold for the impossibility of strong detection, suggesting that strong detection is impossible when $\rho^2 < (1-\varepsilon)d^{-1}$, for some $\varepsilon>0$, even when $d$ is fixed. This threshold, however, is strictly larger than $\rho^\star(d)$, for $d<4$, and thus improves upon  \cite{HuleihelElimelech}. For fixed $d\geq 4$, the threshold $\rho^{\star}$ is strictly larger then $d^{-1}$, and our bound does not improves upon the one in \cite{HuleihelElimelech}. For $k=n$, we only slightly improve existing results; here, our main contribution is our proof technique, which applies for all asymptotic regimes at once, while in \cite{HuleihelElimelech} a separate analysis for each asymptotic regime is required. For general values of $k$, we show in Section~\ref{sec:kLower} that the analysis and techniques in \cite{HuleihelElimelech} do not extend to the partly correlated regime, as it is unclear how to evaluate the obtained expression for the likelihood's second moment (see, Proposition~\ref{prop:nastybound}). This indicates that in addition to the fact that our analysis arguably simpler as compared to \cite{HuleihelElimelech}, our new technique is also essential for analyzing more complex scenarios, which seem quite complicated to analyze using the standard tools.

\vspace{0.2cm}
\noindent\textbf{Upper bounds.} Next, we move to our upper bounds. From the algorithmic point of view, in the high-dimensional case, we propose the following almost straightforward extension of \eqref{eqn:testcount}, 
\begin{align}
    \phi_{\s{count}}(\s{X},\s{Y})\triangleq\Ind\ppp{\sum_{i,j=1}^n\bar{\calI}(X_i,Y_j)\geq \frac{1}{2} k\calP_{\rho,d}},\label{eqn:testcount_high}
\end{align}
where $\tau_{\s{count}}\in\mathbb{R}$, and
\begin{align}
    \bar\calI(X_i,Y_j)\triangleq\Ind\ppp{\frac{1}{d}\sum_{\ell=1}^d\calL_{\s{I}}(X_{i\ell},Y_{j\ell})\geq\tau_{\s{count}}},
\end{align} 
and finally,
\begin{subequations}
\begin{align}
    \calP_{\rho,d}&\triangleq\pr_{P_{XY}^{\otimes d}}\pp{\sum_{\ell=1}^d\calL_{\s{I}}(A_\ell,B_\ell)\geq d\cdot\tau_{\s{count}}},\\
    \calQ_{\rho,d}&\triangleq\pr_{Q_{XY}^{\otimes d}}\pp{\sum_{\ell=1}^d\calL_{\s{I}}(A_\ell,B_\ell)\geq d\cdot\tau_{\s{count}}},
\end{align}
\end{subequations}
where we recall that $Q_{XY} = \calN(\mathbf{0},\Sigma_0)$ and $P_{XY} = \calN(\mathbf{0},\Sigma_\rho)$. We can then show the following result. Applying Chernoff's bound on these tails yield that if the threshold $\tau_{\s{count}}$ is such that $\tau_{\s{count}}\in(-d_{\s{KL}}(Q_{XY}||P_{XY}),d_{\s{KL}}(P_{XY}||Q_{XY}))$, then
\begin{subequations}\label{eqn:chen_high}
\begin{align}
    \calQ_{\rho,d}&\leq \exp\pp{- d\cdot E_Q(\tau_{\s{count}})},\label{eqn:chen_high1}\\
    \calP_{\rho,d}&\geq 1-\exp\pp{-d\cdot E_P(\tau_{\s{count}})},\label{eqn:chen_high2}
\end{align}
where $E_Q$ and $E_P$ are defined in \eqref{eq:chernoff}.
\end{subequations}
\begin{theorem}[Count test strong detection]\label{thm:IndSum}
Fix $d\in\mathbb{N}$. Suppose there is a threshold $\tau_{\s{count}}\in(-d_{\s{KL}}(Q_{XY}||P_{XY}),d_{\s{KL}}(P_{XY}||Q_{XY}))$ with
\begin{subequations}\label{eqn:CountcondEx}
	\begin{align}
    E_Q(\tau_{\s{count}}) &= \omega\p{d^{-1}\log \frac{n^2}{k}},\label{eqn:Countcond1Ex}\\
    E_P(\tau_{\s{count}}) &= \omega(k^{-1}d^{-1}).\label{eqn:Countcond2Ex}
\end{align}
\end{subequations}
Then, $\s{R}(\phi_{\s{count}})\to0$, as $n\to\infty$. In particular, $\s{R}(\phi_{\s{count}})\to 0$, as $n\to\infty$, if $\rho^2 = 1-o((n^2/k)^{-4/d})$. 
\end{theorem}

As mentioned right above Theorem~\ref{thm:upper}, the seemingly similar normalized counting procedure in \cite{HuleihelElimelech}, excludes the $d=1$ case, and in fact requires that $d\geq d_0$, for some fixed $d_0\in\mathbb{N}$. Our result above holds for any natural $d\geq1$. Also, interestingly, the threshold $\rho^2 = 1-o(n^{-4/d})$ coincides with the threshold for the recovery problem \cite{pmlr-v89-dai19b}, achieved by the exhaustive maximum-likelihood estimator, while the count test is clearly efficient. 

The count test is not optimal. Consider the regime where $d\to\infty$. In this case, when $d=\omega(\log n)$, we have $\s{R}(\phi_{\s{count}})\to0$, if $d\rho^2 = \omega(\log (n^2/k)) = \omega(\log n)$, almost independently of the value of $k$. The following simple test, however, exhibits better statistical guarantees in the regime where $k = \Omega(n/\sqrt{\log n})$. Define,
\begin{align}
\phi_{\s{sum}}(\s{X},\s{Y})\triangleq\Ind\ppp{\s{sign}(\rho)\sum_{i,j=1}^nX_i^TY_j>\frac{kd|\rho|}{2}}.\label{eqn:sumtest}
\end{align}
Intuitively, $\phi_{\s{sum}}$ sums all $n^2$ possible inner products of feature pairs, and compares this sum to the threshold $\frac{kd|\rho|}{2}$, which is simply the average between the expected values of the sum under the null and alternative hypotheses. The idea behind this test is that under the alternative hypothesis, where the databases are correlated, this sum is larger as compared to the null hypothesis, where the databases are uncorrelated. Note that this test is by no means novel, and was proposed and analyzed for the case where $k=n$ already in \cite{nazer2022detecting}. Here, we analyze its performance for any value of $k$, and use a simpler first and second moment analysis as compared to \cite{nazer2022detecting}. We have the following result.
\begin{theorem}\label{th:sumtest}
Consider the test in \eqref{eqn:sumtest}. Then, $\s{R}(\phi_{\s{sum}})\to0$, as $d\to\infty$, if $\rho^2=\omega(\frac{n^2}{dk^2})$.
\end{theorem}

Thus, if for example $\frac{k}{n}\to \kappa\in(0,1]$, then we see that $\s{R}(\phi_{\s{sum}})\to0$ if $\kappa^2d\rho^2 = \omega(1)$, strictly better than the guarantees we obtained for the count test. When $\omega(1)=d = O(\log n)$ the count test is successful only if (at best) $\rho^2 = 1-\Omega(1)$, while the sum test allows for vanishing correlations as $\kappa^2d\rho^2 = \omega(1)$. When $d$ is fixed, however, the sum test is completely useless, and the count test is useful. Also, in the regime where $k = O(n/\sqrt{\log n})$, it can be seen that the count test have better statistical guarantees compared to the sum test

Next, we consider weak detection. As mentioned right above Theorem~\ref{thm:upper_comp}, by definition, any test that achieves strong detection achieves weak detection automatically. As so, the count and sum tests above achieve weak detection under the same conditions stated in Theorems~\ref{thm:IndSum} and \ref{th:sumtest}. Nonetheless, as for the $d=1$ case, in the regime where  $d$ is fixed, we can obtain stronger performance guarantees under the weak detection criterion. Specifically, consider the (straightforward) generalization of the comparison test in \eqref{eqn:testcomp}. We define,
\begin{align}
    \phi_{\s{comp}}(\s{X},\s{Y})\triangleq\Ind\ppp{\abs{\sum_{i,j}(X_{ij}-Y_{ij})}\leq\theta},\label{eqn:testcomp_high}
\end{align}
if $\rho\in(0,1]$, and we flip the direction of the inequality in \eqref{eqn:testcomp_high} if $\rho\in[-1,0)$, and $\theta\in\mathbb{R}_+$. Let $\s{G}\sim\calN(0,1)$ and $\s{G}'\sim\calN\p{0,1-\frac{k}{n}|\rho|}$. We define the threshold $\theta$ as the value for which 
\begin{align}
    d_{\s{TV}}\p{\calN(0,1),\calN\p{0,1-\frac{k}{n}|\rho|}}= \pr\p{|\s{G}|\geq\frac{\theta}{\sqrt{2nd}}}-\pr\p{|\s{G}'|\geq \frac{\theta}{\sqrt{2nd}}}.\label{def:thetaTV}
\end{align}
Again, such a value exists by the definition of the total-variation distance for centered Gaussian random variables (see, e.g., \cite[pg. 10]{devroye2023total}). We then have the following result.
\begin{theorem}[Comparison test weak detection]\label{thm:upper_comp_high}
Consider the detection problem in \eqref{eqn:decproblem}, and the comparison test in \eqref{eqn:testcomp_high}, with $\theta$ given by \eqref{eq:thetavalues}. If $\rho^2 = \Omega(1)$ and $k=\Theta(n)$, then $\lim_{n\to\infty}\s{R}(\phi_{\s{comp}})<1$, for any $d\geq1$.
\end{theorem}

Comparing Theorems~\ref{thm:IndSum} and \ref{thm:upper_comp_high}, we see that when $d$ is fixed, then weak detection is possible if $\rho^2$ is of order constant, while for strong detection our test requires the correlation to converge to unity sufficiently fast. On the other hand, if $d\to\infty$, then the count and sum test clearly have better performance guarantees. Finally, we mention here that Theorem~\ref{thm:upper_comp_high} holds for any natural $d\geq1$, while the weak detection guarantee (achieved by the sum test) in \cite{nazer2022detecting,HuleihelElimelech} holds under the assumption that $d> 60\log 2$. For clarity, we emphasize the detection criterion and asymptotic regime where each one the test above are most relevant/operate the best:
\begin{itemize}
    \item Count test: strong detection and $d\in\mathbb{N}$.
    \item Comparison test: weak detection and $d\in\mathbb{N}$.
    \item Sum test: strong detection and $d\to\infty$.
\end{itemize}

\begin{table}[t!]
\small
\begin{center}
\renewcommand{\arraystretch}{2}
\begin{tabular}{ |p{2.7cm}||p{2cm}|p{2.85cm}|p{3.6cm}| p{2.8cm}|}
 \hline
   & \multicolumn{2}{|c|}{\textbf{Weak Detection}} &  \multicolumn{2}{|c|}{\textbf{Strong Detection}}  \\
 \hline
  \textbf{Asymptotics} & \textbf{Possible} & \textbf{Impossible} &\textbf{Possible}& \textbf{Impossible}\\
 \hline 
 $n,d\to \infty$, $k=\Theta(n)$   & $\underset{\text{Theorem~\ref{th:sumtest}, \cite{nazer2022detecting}}}{\Omega\parenv*{d^{-1}}} $    & $\underset{\text{Theorem~\ref{th:lowerWeak2},\cite{HuleihelElimelech}}}{o\parenv*{d^{-1}} }$ &  $\underset{\text{Theorem~\ref{th:lowerStrong2},\cite{nazer2022detecting}}}{\omega\parenv*{d^{-1}}}$    & $\underset{\text{Theorem~\ref{th:lowerStrong2},\cite{HuleihelElimelech}}}{(1-\Omega(1))\cdot d^{-1} }$\\ \hline

 $d\to \infty$, $n$ constant, $k=\Theta(n)$  & $\underset{\text{Theorem~\ref{th:sumtest}, \cite{nazer2022detecting}}}{\Omega\parenv*{d^{-1}}} $     & $\underset{\text{Theorem~\ref{th:lowerWeak2},\cite{HuleihelElimelech}}}{o\parenv*{d^{-1}} }$ &  $\underset{\text{Theorem~\ref{th:sumtest},\cite{nazer2022detecting}}}{\omega\parenv*{d^{-1}}}$    & $\underset{\text{\cite{HuleihelElimelech}}}{O\parenv*{d^{-1}} }$\\ \hline

 $n\to\infty$, $d$ constant, $k=\Theta(n)$ &$\underset{\text{Theorem~\ref{thm:upper_comp_high}}}{\Omega(1)}$ & $\underset{\text{Theorem~\ref{th:lowerWeak2},\cite{HuleihelElimelech}}
 }{o\parenv*{1}}$&  $\underset{\text{Theorem~\ref{thm:upper_comp},\cite{HuleihelElimelech}}}{1-o\p{n^{-\frac{2}{d-1}}\wedge n^{-\frac{4}{d}}}^*}$ &$\underset{\text{Theorem~\ref{th:lowerStrong2},\cite{HuleihelElimelech}}}{\min\p{\rho^\star(d),d^{-1}}}$\\ \hline

    $n\to \infty ,d$ constant, $k=O(\log(n))$ &$\underset{\text{Theorem~\ref{thm:IndSum}}}{ 1-\s{C} (\frac{n^2}{k})^{-\frac{d}{4}}}$ &  $\underset{\text{Theorem~\ref{th:lowerWeak2}}}{d^{-1}[1- \omega((\frac{k}{n})^{\frac{2}{k}})]}$ &  
 $\underset{\text{Theorem~\ref{thm:IndSum}}}{ 1-o\p{(n^2/k)^{-d/4}}}$  & $\underset{\text{Theorem~\ref{th:lowerStrong2}}}{d^{-1}[1- \Omega((\frac{k}{n})^{\frac{2}{k}})]}$\\
\hline
\end{tabular}
\caption{A summary of our bounds on $\rho^2$, for weak and strong detection, as a function of the asymptotic regime. The one marked with $*$ is for $d>1$, where for $d=1$ only $n^{-d/4}$ is defined. The referenced papers contains result matching those presented for the special case $k=n$.}
\label{tab:exact23}
\end{center}
\end{table}

\vspace{0.2cm}
\noindent\textbf{Discussion.} In the regime where $k=o(n)$, our results shows that the detection problem gets statistically harder in the sense that strong detection becomes impossible even if $\rho^2<(1-\epsilon(n,k))d^{-1}$, where $\epsilon(n,k)\to0$, as $n\to\infty$. For example, if $k=O(\log n)$, then
\begin{align}
   \prod_{i=1}^k\frac{1}{1-(d\rho^2)^i}&\leq \pp{1+\frac{d\rho^2}{1-d\rho^2}}^k,
\end{align}
and so \eqref{eqn:strongDec2} is satisfied if
\begin{align}
    \rho^2<\frac{1}{d}\pp{1-\p{\s{C}\frac{k^2}{n^2}}^{\frac{1}{k}}},
\end{align}
for some $\s{C}>0$. We note that 
\begin{equation}\label{eq:klogn}
    \p{\frac{k}{n}}^{\frac{2}{k}}=\exp\p{-\Omega\p{\frac{n}{\log(n)
}}}=o(1),
\end{equation} which proves that the problem is \textit{strictly} harder. To see that, consider, for example, the one-dimensional case $d=1$. For any $\varepsilon>0$, by \eqref{eqn:strongDec2} and \eqref{eq:klogn} we have that detection is impossible if 
\begin{equation}
   1-n^{-(4+\varepsilon)} < \rho^2 \leq 1-\exp\p{-\Omega\p{\frac{n}{\log(n)
}}},
\end{equation}
where such $\rho$ exist for sufficiently large $n$ since $e^{n/\log(n)}$ decays faster than $n^{-\alpha}$, for all $\alpha$. In particular, there exists a sequence $\rho_n$ that satisfies the strong detection condition of the count test in Theorem~\ref{thm:upper} for $k=n$, for which strong detection is theoretically impossible when $k=O(\log n)$. The discussion above is similar for the weak detection criterion as well. When $k=\Theta(n)$, our lower and upper bounds suggest that the detection problem is roughly statistically the same at least asymptotically. %On the other hand, we have the case where $k=\Theta(n)$, in which, the statistical hardness of the problem is the same as in the fully correlated case. Indeed, when the the databases are only partly correlated, the problem is only harder, and in particular, the impossibility thresholds for the fully correlated case lower bounds those of the partly correlated case. For the upper bounds,  first consider the case where $d$ is constant. note that the test of Theorem~\ref{thm:upper_comp_high} for weak detection applies under with the same threshold $\rho^2=\Omega(1)$ as in the fully correlated case. For strong detection the, the threshold for the test of Theorem~\ref{thm:IndSum}, $\rho=1-o((n^2/k)^{-d/4})$ is equivalent to the fully correlated case as well as $n^2/k=\Theta(n)$. In the regime where $d\to \infty$ (and general $n$), a similar argument shows that the threshold of Theorem~\ref{th:sumtest} remains the same as in the fully correlated case. 
Finally, for general values of $k$, our lower and upper bounds are inconclusive on the exact threshold at which a phase transition for weak/strong detection occurs. However, as shown above, the lower bounds in Theorems~\ref{th:lowerStrong2} and \ref{th:lowerWeak2} show that in certain regimes, the number $k$ of correlated pairs has a non-trivial influence on the statistical hardness of the problem. This could only be proved using our new method for analysing the likelihood's second moment. We summarize the above results in Table~\ref{tab:exact23}, next to previously known bounds in the literature. 

%In fact, in Section~\ref{sec:kLower}, we show that the analysis and techniques in \cite{HuleihelElimelech} do not extend to the partly correlated regime, as it is unclear how to evaluate the obtained expression for the likelihood's second moment (see, Proposition~\ref{prop:nastybound}). This indicates that in addition to the fact that our analysis arguably simpler as compared to \cite{HuleihelElimelech}, our new technique is also essential for analyzing more complicated scenarios, which seem quite complicated to analyze using the standard tools.}

\subsection{Lower bounds}\label{sec:kLower}
As in the proof for the $d=1$ case, throughout this entire section, with some abuse of notation, we use lower-case Greek letters (such as, $\alpha$) for integer matrices, and the dimensions of these matrices are mentioned of understood by context. The prove strategy of Theorems~\ref{th:lowerStrong2}--\ref{th:lowerWeak2} is the same as that of Theorems~\ref{th:lowerStrong}--\ref{th:lowerWeak}. We thereby repeat the exact same steps described in the proof outline presented in Section~\ref{sec:outline}. However, we would like to emphasize that despite this similarity, a more refined analysis is required for the more complicated generalized multidimensional partly-correlated scenario. Note that \eqref{eq:dTV} and \eqref{eq:CauchyS} are true regardless of the structure of our problem and therefore hold for the generalized model as well. Accordingly, we follow the same approach of evaluating the second moment of the likelihood ratio by representing it as a linear combination of basis functions, and then use Parseval's identity. We consider the Hilbert space $L^2(\calH_0)$ of all random variables over $\R^{n\times d}\times \R^{n\times d}$, with the inner product defined in \eqref{eq:innerPs}. As in the one-dimensional case, the set of multivariate Hermite polynomials $\ppp{H_{\alpha,\beta}}_{\alpha,\beta}$ form an orthonormal basis, but now $\alpha,\beta\in \N^{n\times d}$ are $n\times d$ coefficient matrices (rather than vectors), and 
\begin{align}
    H_{\alpha,\beta}(\s{X},\s{Y})=\prod_{i=1}^n\prod_{j=1}^d h_{\alpha_{i,j}}(\s{X}_{i,j})\prod_{i=1}^n\prod_{j=1}^d h_{\beta_{i,j}}(\s{Y}_{i,j}).
\end{align}
For an $n\times d$ integer matrix $\alpha$ define $|\alpha|=\sum_{i,j}\alpha_{i,j}$ to be its $L^1$ norm, and $\s{RS}(\alpha)$ to be the row support of $\alpha$, that is, the set of indices of the non-zero rows of $\alpha$,
\begin{align}
    \s{RS}(\alpha)=\ppp{1\leq i \leq n ~|~ (\alpha_{i,1},\dots,\alpha_{i,d})\neq \mathbf{0}}.
\end{align}
We also generalize the definition of the equivalence relation we defined in Section~\ref{sec:LBPOL} for integer vectors to integer matrices. We say that $\alpha\equiv \beta$ if there exists a permutation $\sigma\in \S_n$ such that $\beta$ is obtained by permuting the rows of $\alpha$ according to $\sigma$. We denote the equivalence class of $\alpha$ by $[\alpha]$. Now, using Parseval's identity, we obtain,
\begin{align}
\E_{\calH_0}\pp{\calL(\s{X},\s{Y})^2}&=\norm{\calL(\s{X},\s{Y})}^2_{\calH_0}\\
    &=\sum_{m=0}^\infty \sum_{\substack{(\alpha,\beta)\in (\N^{n\times d})^2\\ |\alpha|+|\beta|=m }}\innerP{H_{\alpha,\beta}(\s{X},\s{Y}),\calL(\s{X},\s{Y})}_{\calH_0}^2. \label{eq:parsevalD}
\end{align}
\sloppy
Our goal now is to give an exact expression for the orthogonal projection coefficients $\innerP{H_{\alpha,\beta}(\s{X},\s{Y}),\calL(\s{X},\s{Y})}_{\calH_0}$. This is done in a similar fashion to Lemma~\ref{lem:OrthCeoff}.
\begin{lemma}
    For any $\alpha,\beta\in \N^{n\times d}$, 
    \begin{align}
        \innerP{H_{\alpha,\beta}(\s{X},\s{Y}),\calL(\s{X},\s{Y})}_{\calH_0}&=\rho^{|\alpha|}\cdot \P[\sigma(\beta)=\alpha]\cdot \P[\s{RS}(\alpha)\subseteq \calK]\\
        & = \Ind\ppp{\alpha\equiv\beta, |\s{RS}(\alpha)|\leq k}\cdot \frac{\rho^{|\alpha|}}{\abs{[\alpha]}} \prod_{i=0}^{\abs{\s{RS}(\alpha)-1}}\p{\frac{k-i}{n-i}}.
    \end{align}
    \label{lem:OrthCeoffD}
\end{lemma}
\begin{proof}
    We start the proof by repeating the steps of the proof of Lemma~\ref{lem:OrthCeoff}. Given $\sigma$ and $\calK$, the pairs $(\s{X_{i,j}},\s{Y}_{\sigma(i),j})_{i,j}$  are statistically independent random variables, and thus,
      \begin{align}
        \innerP{H_{\alpha,\beta}(\s{X},\s{Y}),\calL(\s{X},\s{Y})}_{\calH_0}
        & =\E_{\calH_0}\pp{H_{\alpha,\beta}(\s{X},\s{Y})\frac{\mathrm{d}\P_{\calH_1}}{\mathrm{d}\P_{\calH_0}}}\\
        &=\E_{\calH_1}\pp{H_{\alpha,\beta}(\s{X},\s{Y})}\\
        &=\E_{\sigma,\calK}\pp{\E_{\calH_1|\sigma,\calK}\pp{H_{\alpha,\beta}(\s{X},\s{Y})}}\\
        &=\E_{\sigma}\pp{\E_{\calH_1|\sigma,\calK}\pp{\prod_{i=1}^n\prod_{j=1}^d h_{\alpha_{i,j}}(\s{X}_{i,j}) 
        \cdot h_{\beta_{\sigma(i),j}}(\s{Y}_{\sigma(i),j})}}\\
        &=\E_{\sigma,\calK}\pp{\prod_{i=1}^n\prod_{j=1}^d \E_{\calH_1|\sigma,\calK}\pp{h_{\alpha_{i,j}}(\s{X}_{i,j}) 
        \cdot h_{\beta_{\sigma(i),j}}(\s{Y}_{\sigma(i),j})}}. \label{eq:prod1D}
    \end{align} 
    Now, recall that conditioned on $\sigma$ and $\calK$, for all $i\in [n]$, if $i\notin \calK$, then $\s{X}_{i,j}$ and $\s{Y}_{\sigma(i),j}$ are independent. Thus, if $\alpha_{i,j}$ or $\beta_{\sigma(i),j}$ are non-zero, we get
    \begin{align}
          \E_{\calH_1|\sigma,\calK}\pp{h_{\alpha_{i,j}}(\s{X}_{i,j}) 
        \cdot h_{\beta_{\sigma(i),j}}(\s{Y}_{\sigma(i),j})}
        &=\E_{\calH_1|\sigma,\calK}\pp{h_{\alpha_{i,j}}(\s{X}_{i,j})} 
        \cdot \E_{\calH_1|\sigma,\calK}\p{h_{\beta_{\sigma(i),j}}(\s{Y}_{\sigma(i),j})}\\
        &=0,\label{eq:prodcalc}
      \end{align}
      where the last equality follows since distinct Hermite polynomials are orthogonal under $\innerP{\cdot,\cdot}_{\calH_0}$, and therefore 
      \begin{align}
    \E_{\calH_1|\sigma,\calK}\pp{h_{\alpha_{i,j}}(\s{X}_{i,j})}=\innerP{h_{\alpha_{i,j}},h_0}_{\calH_{0}}=0,
      \end{align}
      if $\alpha_{i,j}\neq0$, and 
       \begin{align}
    \E_{\calH_1|\sigma,\calK}\pp{h_{\beta_{\sigma(i),j}}(\s{Y}_{\sigma(i),j})}=\innerP{h_{\beta_{\sigma(i),j}},h_0}_{\calH_{0}}=0,
      \end{align}
       if $\beta_{\sigma(i),j}\neq 0$.
       Thus, on one hand, given $\sigma$ and $\calK$, we obtain that the product 
       \begin{align}
           \prod_{i=1}^n\prod_{j=1}^d \E_{\calH_1|\sigma,\calK}\pp{h_{\alpha_{i,j}}(\s{X}_{i,j}) 
        \cdot h_{\beta_{\sigma(i),j}}(\s{Y}_{\sigma(i),j})}
       \end{align}
       can be non-zero only if $\alpha_{i,j}=\beta_{\sigma(i),j}=0$ for any $i\neq \calK$ and $j\in [d]$. On the other hand, given $\sigma$, $\calK$,  $i\in \calK$ and $j\in [d]$, by the proof of Lemma~\ref{lem:OrthCeoff}, we have 
\begin{align}
\E_{\calH_1|\sigma,\calK}\pp{h_{\alpha_{i,j}}(\s{X}_{i,j}) 
        \cdot h_{\beta_{\sigma(i),j}}(\s{Y}_{\sigma(i),j})}=\rho^{\alpha_{i,j}}\delta\pp{\alpha_{i,j}-\beta_{\sigma(i),j}}.\label{eq:prodcalc2}
      \end{align}
Therefore, combining \eqref{eq:prodcalc} and \eqref{eq:prodcalc2} we have
\begin{align}
           \prod_{i=1}^n\prod_{j=1}^d \E_{\calH_1|\sigma,\calK}\pp{h_{\alpha_{i,j}}(\s{X}_{i,j}) 
        \cdot h_{\beta_{\sigma(i),j}}(\s{Y}_{\sigma(i),j})}&=\rho^{\sum_{i=1}^n\sum_{j=1}^d \alpha_{i,j}}\Ind\ppp{\substack{\alpha_{i,j}=\beta_{\sigma(i),j} \\
        \forall i\in \calK, j\in [d]}}\cap\ppp{\substack{\alpha_{i,j}=\beta_{\sigma(i),j}=0 \\
        \forall i\notin \calK, j\in [d]}}\\
        &=\rho^{|\alpha|}\Ind\ppp{\sigma(\beta)=\alpha}\Ind\ppp{\s{RS}(\alpha)\subseteq \calK}.
       \end{align}
      We now conclude,
       \begin{align}
        \innerP{H_{\alpha,\beta}(\s{X},\s{Y}),\calL(\s{X},\s{Y})}_{\calH_0}     &=\E_{\sigma,\calK}\pp{\prod_{i=1}^n\prod_{j=1}^d \E_{\calH_1|\sigma,\calK}\pp{h_{\alpha_{i,j}}(\s{X}_{i,j}) 
        \cdot h_{\beta_{\sigma(i),j}}(\s{Y}_{\sigma(i),j})}}\\
        & =\E_{\calH_1|\sigma,\calK}\pp{\rho^{|\alpha|}\Ind\ppp{\sigma(\beta)=\alpha}\Ind\ppp{\s{RS}(\alpha)\subseteq \calK}}\\
        & =\rho^{|\alpha|}\P\pp{\sigma(\beta)=\alpha}\P\pp{\s{RS}(\alpha)\subseteq \calK}.
    \end{align} 
        Since $\sigma$ is uniform on $\S_n$, $\sigma(\beta)$ is uniformly distributed over all $[\beta]$, which is the set of matrices obtained by permutations on the rows of $\beta$. Thus, 
        \begin{align}
            \P\pp{\sigma(\beta)=\alpha}=\frac{1}{|[\beta]|}\Ind\ppp{\alpha\equiv \beta}=\frac{1}{|[\alpha]|}\Ind\ppp{\alpha\equiv \beta}.
        \end{align}
        For a fixed $\alpha\equiv\beta$, note that $\P\pp{\s{RS}(\alpha)\subseteq \calK}$ is non-zero only if $|\s{RS}(\alpha)|\leq k$, in which case, since $\calK$ is uniform over $\binom{[n]}{k}$, 
        \begin{align}
            \P\pp{\s{RS}(\alpha)\subseteq \calK}=\frac{\binom{n-|\s{RS}(\alpha)|}{k-|\s{RS}(\alpha)|}}{\binom{n}{k}}=\prod_{i=0}^{|\s{RS}(\alpha)|-1}\frac{k-i}{n-i}. 
        \end{align}
        This completes the proof.
\end{proof}

We carry on with our analysis of $\E_{\calH_0}[\calL(\s{X},\s{Y})^2]$ in \eqref{eq:parsevalD}. Plugging in the orthogonal projection coefficients calculated in Lemma~\ref{lem:OrthCeoffD}, we obtain,
\begin{align}
     \E_{\calH_0}\pp{\calL(\s{X},\s{Y})^2}
    &=\sum_{m=0}^\infty \sum_{\substack{(\alpha,\beta)\in (\N^{n\times d})^2\\ |\alpha|+|\beta|=m }}\innerP{H_{\alpha,\beta}(\s{X},\s{Y}),\calL(\s{X},\s{Y})}_{\calH_0}^2\\
        &=1+\sum_{m=1}^\infty \sum_{\substack{(\alpha,\beta)\in (\N^{n\times d})^2\\ |\alpha|+|\beta|=m }}\innerP{H_{\alpha,\beta}(\s{X},\s{Y}),\calL(\s{X},\s{Y})}_{\calH_0}^2\\
    &=1+\sum_{m=1}^\infty \sum_{\substack{(\alpha,\beta)\in (\N^{n\times d})^2\\ |\alpha|+|\beta|=m }} \Ind\ppp{\alpha\equiv\beta, |\s{RS}(\alpha)|\leq k}\frac{\rho^{2|\alpha|}}{\abs{[\alpha]}^2} \prod_{i=0}^{\abs{\s{RS}(\alpha)-1}}\p{\frac{k-i}{n-i}}^2\\
    &=1+\sum_{m=1}^\infty \sum_{\ell=0}^k\sum_{\substack{\alpha \in\N^{n\times d}\\ |\alpha|=m\\|\s{RS}(\alpha)|=\ell }}\sum_{\beta\equiv\alpha} \frac{\rho^{2m}}{\abs{[\alpha]}^2} \prod_{i=0}^{\ell-1}\p{\frac{k-i}{n-i}}^2\\
    &=1+\sum_{m=1}^\infty \rho^{2m} \sum_{\ell=0}^k \prod_{i=0}^{\ell-1}\p{\frac{k-i}{n-i}}^2\sum_{\substack{\alpha \in\N^{n\times d}\\ |\alpha|=m\\|\s{RS}(\alpha)|=\ell }}\sum_{\beta\equiv\alpha} \frac{1}{\abs{[\alpha]}^2} \\
    &=1+\sum_{m=1}^\infty \rho^{2m} \sum_{\ell=0}^k \prod_{i=0}^{\ell-1}\p{\frac{k-i}{n-i}}^2\sum_{\substack{\alpha \in\N^{n\times d}\\ |\alpha|=m\\|\s{RS}(\alpha)|=\ell }}\frac{1}{\abs{[\alpha]}} \\
    &=1+\sum_{m=1}^\infty \rho^{2m} \sum_{\ell=0}^k \prod_{i=0}^{\ell-1}\p{\frac{k-i}{n-i}}^2\abs{\ppp{[\alpha] ~\bigg| \substack{\alpha\in\N^{n\times d}\\ |\alpha|=m\\|\s{RS}(\alpha)|=\ell }}}.\label{eq:sumwithequiv}
\end{align}
Note that when $k=n$ and $d=1$, \eqref{eq:sumwithequiv} reduces to the expression in \eqref{eq:parsevalDec}, as expected. 

We now turn to estimate the number of equivalence classes $\s[\alpha]$, for $\alpha\in \N^{n\times d}$, $|\alpha|=m$ and $|\s{RS}(\alpha)|=\ell$, for a fixed $\ell\leq k$. To that end, as in the one-dimensional case, we notice to a correspondence between equivalence classes and $d$-dimensional integer distribution functions with certain properties. Specifically, for $\alpha \in \N^{n\times d}$, consider the function $p_{\alpha}:\N^{d}\to \N$, such that $p_\alpha(v)$ is the number of rows of $\alpha$, which are equal to $v\in\N^d$. We also note that $|\alpha|=m$ and $|\s{RS}(\alpha)|=\ell$ if and only if 
\begin{align}
    \sum_{v\in\N^d}p_\alpha(v)=\ell \quad \text{and}\quad \sum_{v\in\N^d}|v|p_\alpha(v)=m,\label{eq:distcond}
\end{align}
where $|v|=\sum_{i=1}^d v_i$. As in the one-dimensional case, we have that $\alpha\equiv \beta$ if and only if $p_\alpha=p_\beta$. Thus, we wish to bound the number of integer distribution functions 
which satisfy \eqref{eq:distcond}. 

\begin{lemma}\label{lem:ICQ}
For fixed $\ell\leq k$ and $m\in \N$, we have
\begin{align}
    \sum_{p_1\in \s{Par}(m,\ell)}\prod_{q\in \N}\frac{1}{p_1(q)!}\binom{q+d-1}{d-1}^{p_1(q)}&\leq \abs{ \ppp{p:\N^d\to \N ~|~ p \text{ satisfies \eqref{eq:distcond}}}},\label{eqn:low1}
    \end{align}
    where $p_1$ is a partition of the integer $m$ to exactly $\ell$ parts,\footnote{It is convenient to think about $p_1$ as a one-dimensional integer distribution function, such that $p_1(q)$ is the number of parts of size $q$ in a given partition.} and,
    \begin{align}
    \abs{ \ppp{p:\N^d\to \N ~|~ p \text{ satisfies \eqref{eq:distcond}}}}&\leq\sum_{p_1\in \s{Par}(m,\ell)}\prod_{q\in \N}\binom{q+d-1}{d-1}^{p_1(q)}\label{eqn:upp1}\\
    &\leq d^m|\s{Par}(m,\ell)|.\label{eqn:upp2}
\end{align}
\end{lemma}

\begin{proof}
     We note that any $d$-dimensional distribution function $p$ which satisfies \eqref{eq:distcond} corresponds uniquely to an unordered list of $\ell$ non-zero vectors in $\N^d$, such that the sum over all entries of all vectors together is $m$. For a given vector $v\in\N^d$, $p(v)$ represents the number of occurrences of the vector $v$ in this list. Thus, in order to prove the desired bounds on $\abs{ \ppp{p:\N^d\to \N ~|~ p \text{ satisfies \eqref{eq:distcond}}}}$, we will bound the number of such unordered lists.

\noindent\textbf{Proof of \eqref{eqn:upp1}.} We upper bound the number of the unordered lists described above by using a procedure that generates all possible lists, and then bounding the number of distinct lists generated by this procedure. We first go over all possible norms of vectors $v$ in that list. Since the summation of these weights is exactly $m$, we should go over all partitions of the integer $m$ to exactly $\ell$ elements. Having the norms of the vectors in the list fixed, we upper bound the number of ways to choose integer vectors with these exact norms. For a partition part with weight $q$, the number of vectors in $\N^d$ with weight $q$ is exactly the same as the number of ways to distribute $q$ elements into $d$ cells, which is known to be equal to $\binom{q+d-1}{d-1}$. For a fixed partition $p_1$, there are $p_1(q)$ parts with weight $q$, and therefore, there are at most $\binom{q+d-1}{d-1}^{p_1(q)}$ ways to generate $d$-dimensional integer vectors with such norm. Taking into account all of the distinct parts of the partition $p_1$, we upper bound the number ways to generate appropriate $d$-dimensional integer vectors by,
\begin{align}\label{eq:itemsincells}
    \prod_{q\in \N}\binom{q+d-1}{d-1}^{p_1(q)}.
\end{align}
This concludes the first upper-bound.

\noindent\textbf{Proof of \eqref{eqn:upp2}.} For the second upper bound, we note that the factor $\binom{q+d-1}{d-1}$, which is the number of ways to distribute $q$ elements into $d$ cells, can be upper bounded by $d^q$. If each one of the $q$ elements is a associated with a distinct item, they have $d$ optional cells each, and therefore, together, the $q$ elements have $d^q$ optional distributions along the $d$ cells. Hence, 
\begin{align}
     \ppp{p:\N^d\to \N ~|~ p \text{ satisfies \eqref{eq:distcond}}}&\leq\sum_{p_1\in \s{Par}(m,\ell)}\prod_{q\in \N}\binom{q+d-1}{d-1}^{p_1(q)}\\
     &\leq \sum_{p_1\in \s{Par}(m,\ell)}\prod_{q\in \N}d^{q\cdot p_1(q)}\\
     &\leq \sum_{p_1\in \s{Par}(m,\ell)}d^{\sum_{q\in \N}q\cdot p_1(q)}\\
     &= \sum_{p_1\in \s{Par}(m,\ell)}d^{m}\\
     &=\abs{\s{Par}(m,\ell)}d^{m}.
\end{align}

\noindent\textbf{Proof of \eqref{eqn:low1}.} For the lower bound, we note that for a fixed $q$, the expression in \eqref{eq:itemsincells} may over-count the number of possible ways to choose $p_1(q)$ \textit{unordered} $d$-dimensional vectors whose norms are $q$, because different orderings of such vectors are counted multiple times. Thus, we lower bound the number of possible ways to choose $p_1(q)$ \textit{unordered} $d$-dimensional vectors with norms equal to $q$, by dividing by $p_1(q)!$, the maximal number of multiple counts for each set of $p_1(q)$ vectors. The rest of the analysis is carried out similarly to the upper bound.
\end{proof}

We are now ready to prove our lower bounds.

\begin{proof}[Proofs of Theorems~\ref{th:lowerStrong2}--\ref{th:lowerWeak2}]

    We continue from \eqref{eq:sumwithequiv}, and use Lemma~\ref{lem:ICQ} to obtain,
\begin{align}
    \E_{\calH_0}\pp{\calL(\s{X},\s{Y})^2}&=1+\sum_{m=1}^\infty \rho^{2m} \sum_{\ell=0}^k \prod_{i=0}^{\ell-1}\p{\frac{k-i}{n-i}}^2\abs{\ppp{[\alpha] ~\bigg| \substack{\alpha\in\N^{n\times d}\\ |\alpha|=m\\|\s{RS}(\alpha)|=\ell }}}\\
    &=1+\sum_{m=1}^\infty \rho^{2m} \sum_{\ell=0}^k \prod_{i=0}^{\ell-1}\p{\frac{k-i}{n-i}}^2\abs{\ppp{[\alpha] ~\bigg| \substack{\alpha\in\N^{n\times d}\\ |\alpha|=m\\|\s{RS}(\alpha)|=\ell }}}\label{eq:precise}\\
    &=1+\sum_{m=1}^\infty \rho^{2m} \sum_{\ell=0}^k \prod_{i=0}^{\ell-1}\p{\frac{k-i}{n-i}}^2 d^{m}\abs{\s{Par}(m,\ell)}\\
     &\overset{(a)}{\leq}1+\sum_{m=1}^\infty \p{d\rho^{2}}^m \sum_{\ell=0}^k \p{\frac{k}{n}}^{2\ell} \abs{\s{Par}(m,\ell)}\\
     &\overset{(b)}{\leq}1+\sum_{m=1}^\infty \p{d\rho^{2}}^m \sum_{\ell=0}^k \p{\frac{k}{n}}^{2} \abs{\s{Par}(m,\ell)}\\
     &=1+\p{\frac{k}{n}}^{2} \sum_{m=1}^\infty \p{d\rho^{2}}^m \sum_{\ell=0}^k \abs{\s{Par}
     (m,\ell)}\\
     &=1+\p{\frac{k}{n}}^{2} \sum_{m=1}^\infty \p{d\rho^{2}}^m \abs{\s{Par}(m,\leq_k)},\label{eq:dkfinitebound}
\end{align}
where $(a)$ follows since the function $\frac{k-i}{n-i}$ is monotonically decreasing as a function of $i$, as long as $n>k$, and $(b)$ follows since $|\s{Par(m,0)}|=0$, for any positive integer, and $(k/n)^{2\ell}$ is decreasing with $\ell$.

Now, recall the well-known formula for the generating function of integer partitions to at most $k$ elements (see, for example, \cite[Chapter 1.3]{flajolet2009analytic}),
\begin{align}
    P_{\leq_k}(z)=\sum_{m=0}^\infty |\s{Par}(m,\leq_k)|z^i=\prod_{i=1}^k\frac{1}{1-z^i},
\end{align}
where the infinite sum converges when $|z|< 1$.
Thus, for any $\rho^2< \frac{1}{d}$, the infinite sum in \eqref{eq:dkfinitebound} converges, and we have,
\begin{align}
    \E_{\calH_0}[\calL(\s{X},\s{Y})^2]&\leq 
   1+ \p{\frac{k}{n}}^2\sum_{m=1}^\infty |\s{Par}(m,\leq_k)|(d\rho^2)^m\\
   &\leq 
   1-\p{\frac{k}{n}}^2 +\p{\frac{k}{n}}^2\sum_{m=0}^\infty |\s{Par}(m,\leq_k)|(d\rho^2)^m \\
    &= 1+\p{\frac{k}{n}}^2\p{\prod_{i=1}^k\frac{1}{1-(d\rho^2)^i}-1}.
\end{align}
The proof is concluded by repeating the same exact steps as in the proofs of Theorems~\ref{th:lowerStrong}--\ref{th:lowerWeak}, by using Cauchy-Schwartz inequality \eqref{eq:CauchyS} for the impossibility of weak detection, and Lemma~\ref{lem:ICQ} for the impossibility of strong detection.
\end{proof} 

Let us now compare our proof technique with the one in \cite{nazer2022detecting, HuleihelElimelech, paslev2023testing}, where a straightforward analysis of $\E_{\calH_0}[\calL^2]$ is carried out. Repeating the same steps in \cite{HuleihelElimelech}, while generalizing the analysis for the partly correlated scenario, one can show that  $\E_{\calH_0}[\calL(\s{X},\s{Y})^2]$ can be expressed using the following representation of random permutations cycle enumerators. The proof of the following result appear is relegated to Appendix~\ref{app:cyclesPer}.
\begin{prop}\label{prop:nastybound}
    The second moment of the likelihood ratio for the detection problem defined in \eqref{eqn:decproblem2} is given by 
    \begin{align}
        \E_{\calH_0}[\calL(\s{X},\s{Y})^2] =\bE_{\sigma}\bE_{\calK}\pp{\prod_{\ell=1}^k\frac{1}{(1-\rho^{2\ell})^{d\s{N}_{\ell}(\sigma,\calK)}}}, \label{eq:nastycalc}
    \end{align}
    where $\s{N}_{\ell}(\sigma,\calK)$ is the number of cycles of $\sigma$ whose intersection with the set $[k]\cap \calK$ is of size $\ell$. 
\end{prop}
In the fully correlated scenario where $k=n$, our generalized model reduces to the one studied in \cite{HuleihelElimelech,nazer2022detecting}. In this case, $\s{N}_\ell(\sigma,\calK)$ becomes $N_\ell(\sigma)$, which is the number of cycles of length $\ell$ of a random permutation. The analysis of \eqref{eq:nastycalc} is already involved and requires a delicate case by case  analysis for the different asymptotic regimes of $n,d$, and $\rho$, using a certain Poisson approximation of the random variables $(N_\ell(\sigma))_\ell$. Our method for calculating $\E_{\calH_0}[\calL(\s{X},\s{Y})^2]$, however, is able to recover the results in \cite{HuleihelElimelech}, but for all regimes at once.

In the case where $k<n$, it becomes unclear how to evaluate the expression in \eqref{eq:nastycalc}. Conditioning on the size of the intersection $\calK\cap [k]$, which we denote here by $i$, it is possible to express \eqref{eq:nastycalc} as a weighted average of elements of the form 
\begin{align}
    \E_{\sigma}\pp{\prod_{\ell=1}^k\frac{1}{(1-\rho^{2})^{d\s{N}_\ell(\sigma,[i]
    )}}},
\end{align}
where $\s{N}_\ell(\sigma,[i])$ is the size of the number of cycles of $\sigma$ intersecting the interval $[i]$ in exactly $\ell$ points. It seems likely that a Poisson approximation for the variables $(\s{N}_\ell(\sigma,[i]))_{\ell,i}$ exists, and may be used in order to estimate $\E_{\calH_0}[\calL(\s{X},\s{Y})^2]$. However, proving such a Poisson approximation theorem is highly non-trivial, and has its own intellectual merit. We therefore leave it for future research.

\begin{remark}
    As discussed above, in the fully correlated scenario where $k=n$ and $d\to\infty$, the lower bounds in Theorems~\ref{th:lowerStrong2} and \ref{th:lowerWeak2} recover the sharp threshold $d^{-1}$, proved in \cite{HuleihelElimelech}. When $d$ is fixed, on the other hand, our lower bounds are not tight in general. To wit, recall that it was proved in \cite{HuleihelElimelech} that strong detection is impossible if $\rho<\rho^\star(d)$, and $\rho^\star(d)$ is a certain function of $d$. As it turns out, for $d\geq 4$, we have $\rho^\star(d)>1/d$, and  so as our results in this paper do not improve upon \cite{HuleihelElimelech}.   
    This suggests that a more delicate analysis of \eqref{eq:precise} is required. We suspect that the loose step in our proof, is use of the bound $|\s{Par}{(m,\ell)}|d^m$ on the number of equivalence classes in \eqref{eq:precise}. As proved in Lemma~\ref{lem:ICQ}, a sharper bound is, 
    \begin{align}
         \sum_{p_1\in \s{Par}(m,\ell)}\prod_{q\in \N}\binom{q+d-1}{d-1}^{p_1(q)}.\label{eq:sumpar}
    \end{align}
    The (almost) matching lower bound from Lemma~\ref{lem:ICQ} hints that \eqref{eq:sumpar} might be a strictly better candidate for analysing \eqref{eq:precise}. Unfortunately, evaluating \eqref{eq:sumpar} seems quite complicated, and calls for a fine analysis of high-order statistics of random integer partitions. 
\end{remark}

\subsection{Upper bounds}\label{swubsec:proofUpperCount} 
\vspace{0.2cm} \noindent\textbf{Count test.} We prove now Theorem~\ref{thm:IndSum}, and mention again that Theorem~\ref{thm:upper} is just the special case where $d=1$ and $k=n$. Consider the test in \eqref{eqn:testcount_high}. We start by bounding the Type-I error probability. Markov's inequality implies that
\begin{align}
    \pr_{\calH_0}\p{\phi_{\s{count}}=1} &= \pr_{\calH_0}\p{\sum_{i,j=1}^n\calI(X_i,Y_j)\geq \frac{1}{2} k\calP_{\rho}}\\
    &\leq \frac{2n^2\calQ_{\rho}}{k\calP_{\rho}}.\label{eqn:type1count}
\end{align}
On the other hand, we bound the Type-II error probability as follows. Under $\calH_1$, since our proposed test is invariant to reordering of $\s{X}$ and $\s{Y}$, we may assume without loss of generality that the latent permutation is the identity one, i.e., $\sigma=\s{Id}$. Then, Chebyshev's inequality implies that
\begin{align}
     \pr_{\calH_1}\p{\phi_{\s{count}}=0} &= \pr_{\calH_1}\p{\sum_{i,j=1}^n\calI(X_i,Y_j)< \frac{1}{2} k\calP_{\rho}}\\
     &\leq \pr_{\calH_1}\p{\sum_{i=1}^k\calI(X_i,Y_i)< \frac{1}{2} k\calP_{\rho}}\\
&\leq\frac{4\cdot\s{Var}_{\rho}\p{\sum_{i=1}^k\calI(X_i,Y_i)}}{k^2\calP^2_{d}},
\end{align}
where $\s{Var}_{\rho}$ denotes the variance w.r.t. $P_{XY}$. Noticing that 
\begin{align}
    \s{Var}_{\rho}\p{\sum_{i=1}^k\calI(X_i,Y_i)} &= \sum_{i=1}^k\s{Var}_{\rho}\p{\calI(X_i,Y_i)}\\
    & = k\calP_{\rho}(1-\calP_{\rho}),
\end{align}
we finally obtain,
\begin{align}
     \pr_{\calH_1}\p{\phi_{\s{count}}=0} &\leq\frac{4(1-\calP_{\rho})}{k\calP_{\rho}}\leq\frac{4}{k\calP_{\rho}}.\label{eqn:type2count}
\end{align}
Thus, using \eqref{eqn:chen} we obtain
\begin{align}
    \pr_{\calH_0}\p{\phi_{\s{count}}=1}&\leq \frac{2n^2}{k}\cdot\frac{\exp\pp{-d\cdot E_Q(\tau_{\s{count}})}}{1-\exp\pp{-d\cdot E_P(\tau_{\s{count}})}},
\end{align}
and
\begin{align}
    \pr_{\calH_1}\p{\phi_{\s{count}}=0}&\leq \frac{4}{k\cdot\p{1-\exp\pp{-d\cdot E_P(\tau_{\s{count}})}}}.
\end{align}
For strong detection, we see that the Type-II error probability converges to zero as $n\to\infty$, if $1-e^{- d\cdot E_P(\tau_{\s{count}})}=\omega(k^{-1})$, which is equivalent to $E_P(\tau_{\s{count}}) = \omega(k^{-1}d^{-1})$. Under this condition we see that the Type-I error probability converges to zero if $\frac{n^2}{k}e^{-d\cdot E_Q(\tau_{\s{count}})} = o(1)$, which is equivalent to $E_Q(\tau_{\s{count}}) = \omega(d^{-1}\log \frac{n^2}{k})$.

Finally, we show that for \eqref{eqn:CountcondEx} to hold it is suffice that $\rho^2 = 1-o((n^2/k)^{-4/d})$. We start by calculating $\psi_P$ and $\psi_Q$. The following result is proved in Appendix~\ref{app:logmom}.
\begin{lemma}\label{lem:simplecalc}
For $-\frac{1-\rho^2}{\rho^2}\leq\lambda$ and $|1-\lambda|<\frac{1}{\rho}$, 
    \begin{align}
  \psi_Q(\lambda) %&= -\frac{\lambda-1}{2}\log(1-\rho^2)-\frac{1}{2}\log[1-(1-2\lambda^2)\rho^2]\\
  & = -\frac{\lambda-1}{2}\log(1-\rho^2)-\frac{1}{2}\log[1-(1-\lambda)^2\rho^2],
\end{align}
and for $|\lambda|<1/\rho^2$,
\begin{align}
    \psi_P(\lambda) = -\frac{\lambda}{2}\log(1-\rho^2)-\frac{1}{2}\log(1-\lambda^2\rho^2).
\end{align}
\end{lemma}
\begin{comment}
\begin{align}
    \psi_Q(\lambda) &= \log\bE_Q\pp{\frac{1}{(1-\rho^2)^{\lambda/2}}\exp\p{\lambda\frac{-(x^2+y^2)\rho^2+2 x y\rho}{2(1-\rho^2)}}}\\
    & = -\frac{\lambda}{2}\log(1-\rho^2)+\log\bE_Q\pp{\exp\p{\lambda\frac{-(x^2+y^2)\rho^2+2 x y\rho}{2(1-\rho^2)}}}.
\end{align}
Using the fact that the moment generating function of Gaussian random variable $\s{W}\sim\calN(\mu,\sigma^2)$ squared is given by,
\begin{align}
    \bE [\exp\p{t\s{W}^2}] = \frac{1}{\sqrt{1-2t\sigma^2}} \exp\left(\frac{\mu^2 t}{1-2t\sigma^2}\right),
\end{align}
for $\s{Real}(t\sigma^2)<1/2$, it can be shown that
\begin{align}
    \bE_Q\pp{\exp\p{\lambda\frac{-(x^2+y^2)\rho^2+2 x y\rho}{2(1-\rho^2)}}}  %=\frac{\sqrt{1-\rho^2}}{\sqrt{1-(1-2\lambda)^2\rho^2}}
    = \frac{\sqrt{1-\rho^2}}{\sqrt{1-(1-\lambda)^2\rho^2}},
\end{align}
for %$-\frac{1-\rho^2}{\rho^2}\leq\lambda$ and $\lambda[1-(1-(1-\lambda)\rho^2)^{-1}]<1$ \textcolor{red}{
$-\frac{1-\rho^2}{\rho^2}\leq\lambda$ and $|1-\lambda|<\frac{1}{\rho}$. Thus,
\end{comment}
Using Lemma~\ref{lem:simplecalc} with $\tau_{\s{count}}=0$ we have
\begin{align}
    E_Q(0) &= \sup_{\lambda\in\mathbb{R}}-\psi_Q(\lambda)\geq-\psi_Q(1/2) \\
    &= -\frac{1}{4}\log(1-\rho^2)+\frac{1}{2}\log(1-\rho^2/4).\label{eqn:EQExamp}
\end{align}
Similarly,
\begin{align}
    E_P(0) &= \sup_{\lambda\in\mathbb{R}}-\psi_P(\lambda)\geq-\psi_P(-1/2) \\
    &= -\frac{1}{4}\log(1-\rho^2)+\frac{1}{2}\log(1-\rho^2/4).
\end{align}
Thus, \eqref{eqn:CountcondEx} hold if $\rho^2 = 1-o((n^2/k)^{-4/d})$.

\vspace{0.2cm} \noindent\textbf{Sum test.} For simplicity of notation define $\s{T}(\s{X},\s{Y})\triangleq\sum_{i,j=1}^nX_i^TY_j$. Let us analyze first the Type-II error probability. We have
\begin{align}
    \pr_{\calH_1}\pp{\phi(\s{X},\s{Y})=0}& = \pr_{\calH_1}\pp{\s{T}(\s{X},\s{Y})\leq\tau}\\
    &\leq\frac{\s{Var}_{\calH_1}\p{\s{T}(\s{X},\s{Y})}}{\p{\bE_{\calH_1}[\s{T}(\s{X},\s{Y})]-\tau}^2},
\end{align}
assuming that $\tau\leq \bE_{\calH_1}[\s{T}(\s{X},\s{Y})] = kd\rho$. Let us find the variance of $\s{T}(\s{X},\s{Y})$ under $\calH_1$. Since $\bE_{\calH_1}[\s{T}(\s{X},\s{Y})] = kd\rho$, we are left with the calculation of the second moment of $\s{T}(\s{X},\s{Y})$. 

Without loss of generality, assume that $\sigma^\star=\s{Id}$ and $\calK=[k]$ (as the bound is invariant under averaging over $\sigma$ and $\calK$).  Furthermore, note that for $i\in\mathcal{K}$ we have $ (X_i,Y_i)$ are equally distributed as $(\rho Y_i+\sqrt{1-\rho^2}Z_i,Y_i)$, where $Z_i\sim \calN(0_d,\mathbf{I}_d)$, and 
 $\ppp{Z_i}_i$  are i.i.d, independent of $\ppp{Y_i}_i$. Thus,
\begin{align}
    \s{T}(\s{X},\s{Y}) &= \sum_{i=1}^{k}\sum_{j=1}^n (\rho Y_i+\sqrt{1-\rho^2}Z_i)^TY_j+\sum_{i=k+1}^{n}\sum_{j=1}^n Z_i^T\cdot Y_j\\
    & = \rho\cdot\sum_{i=1}^{k}\sum_{j=1}^nY_i^TY_j+\sqrt{1-\rho^2}\sum_{i=1}^{k}\sum_{j=1}^nZ_i^TY_j+\sum_{i=k+1}^{n}\sum_{j=1}^nZ_i^TY_j\\
    & = \rho\cdot\sum_{i,j=1}^{k}Y_i^TY_j+\rho\cdot\sum_{i=1}^{k}\sum_{j=k+1}^nY_i^TY_j+\sqrt{1-\rho^2}\sum_{i=1}^{k}\sum_{j=1}^nZ_i^TY_j+\sum_{i=k+1}^{n}\sum_{j=1}^nZ_i^TY_j\\    
    & = \rho\norm{\bar{Y}_k}_2^2+\rho\cdot\bar{Y}_k^T\bar{Y}^{(k+1)}+\sqrt{1-\rho^2}\bar{Z}_k^T\bar{Y}+(\bar{Z}^{k+1})^T\bar{Y},
\end{align}
where we have defined $\bar{Y}\triangleq\sum_{i=1}^n Y_i$, $
\bar{Y}_k\triangleq\sum_{i=1}^kY_i$, $\bar{Y}^{(k+1)}\triangleq\sum_{i=k+1}^nY_i$, $
\bar{Z}_k = \sum_{i=1}^k Z_i$, and $
\bar{Z}^{k+1} = \sum_{i=k+1}^n Z_i$. Note that 
$\bar{Y}\sim $ , $
\calN(0_d,n\cdot\mathbf{I}_d)$$ , $
$\bar{Y_k}\sim$ , $
\calN(0_d,k\cdot\mathbf{I}_d)$$ , $
$\bar{Z}_k \sim $ , $
\calN(0_d,k\cdot\mathbf{I}_d)$, and finally $ 
$$\bar{Z}^{k+1} \sim$
$\calN(0_d,(n-k)\cdot\mathbf{I}_d)$. We also note that 
\begin{align}
    (\bar{Z}_k,\bar{Z}_2)\indep (\bar{Y},\bar{Y}_k, \bar{Y}^{k+1}), \text{ and } \bar{Y}_k\indep \bar{Y}^{k+1}.\label{eq:indep}
\end{align} 
Thus, since all vectors have zero mean, the expected value of the crossing terms in $\E[\s{T}(\s{X},\s{Y})^2]$ are nullified, and as so,
\begin{align}
    \E_{\calH_1}\pp{\s{T}(\s{X},\s{Y})^2} &= \rho^2\bE\pp{\norm{\bar{Y}_k}_2^4}+\rho^2\bE\pp{\p{\bar{Y}_k^T\bar{Y}^{(k+1)}}^2}+(1-\rho^2)\bE\pp{\p{\bar{Z}_k^T\bar{Y}}^2}\nonumber\\
    &\quad+\bE\pp{\p{(\bar{Z}^{k+1})^T\bar{Y}}^2}.
\end{align}
By the independence structure in \eqref{eq:indep}, we get,
\begin{align}
   \bE\pp{\p{\bar{Z}_k^T\bar{Y}}^2} &=\bE[\bar{Y}^T\bar{Z}_k\bar{Z}_k^T\bar{Y}]\\
   &=\E\pp{\left.\E\pp{\bar{Y}^T\bar{Z}_k\bar{Z}_k^T\bar{Y} \right| \bar{Y}}}\\
   &= \E\pp{\bar{Y}^T\cdot \E\pp{\left.\bar{Z}_k\bar{Z}_k^T \right| \bar{Y}}\cdot \bar{Y}}\\
   &= \E\pp{\bar{Y}^T\E\pp{\bar{Z}_k\bar{Z}_k^T }\bar{Y}}\\
   &= k\bE\pp{\norm{\bar{Y}}_2^2} = knd.
\end{align}
A similar calculation shows that $\bE\pp{((\bar{Z}^{k+1})^T\bar{Y})^2}= (n-k)\bE\pp{\norm{\bar{Y}}_2^2} = (n-k)nd$. Next, we have,
\begin{align}
    \bE\pp{\p{\bar{Y}_k^T\bar{Y}^{(k+1)}}^2} &= \bE\pp{\bar{Y}_k^T\bar{Y}^{(k+1)}(\bar{Y}^{(k+1)})^T\bar{Y}_k} \nonumber\\
    &= (n-k)\bE\pp{\norm{\bar{Y}_k}_2^2 }.
\end{align}
Recall that $\bar{Y_k}\sim \calN(0_d,k\cdot\mathbf{I}_d)$. Denoting the elements of $\bar{Y}$ by $\bar{Y} =(\bar{Y}_1,\bar{Y}_2,\ldots,\bar{Y}_d)$, we have that for each $1\leq i\leq d$, $\bar{Y}_i\sim \calN(0,k)$. Thus,
\begin{align}
    \bE\norm{\bar{Y_k}}_2^2  = \bE\pp{\bar{Y}_1^2+\bar{Y}_2^2+\ldots+\bar{Y}_d^2} = kd,
\end{align}
implying that
\begin{align}
    \bE\pp{\p{\bar{Y}_k^T\bar{Y}^{(k+1)}}^2} = (n-k)kd.
\end{align}
Finally, 
\begin{align}
    \bE\pp{\norm{\bar{Y_k}}_2^4} &= \bE\pp{\p{\sum_{i=1}^d\bar{Y}_i^2}^2}\\
    & = \bE\pp{\sum_{i=1}^d\bar{Y}_i^4+\sum_{i\neq j}\bar{Y}_i^2\bar{Y}_j^2}\\
    & = 3dk^2+d(d-1)k^2,
\end{align}
where we have used the fact that for $W\sim \calN(0,\nu^2)$, it holds $\bE W^4 = 3\sigma^4$. Combining the above we finally obtain
\begin{align}
    \s{Var}_{\calH_1}\p{\s{T}(\s{X},\s{Y})}&= %\bE_{\calH_1}\pp{\s{T}(\s{X},\s{Y})}^2-\p{\bE_{\calH_1}\pp{\s{T}(\s{X},\s{Y})}}^2\\
    %& = \rho^2\pp{3dk^2+d(d-1)k^2+dk(n-k)}+(1-\rho^2)knd\nonumber\\
    %&\ \ \ +(n-k)nd-k^2d^2\rho^2\\
    \rho^2\p{dk^2+dkn}+(1-\rho^2)knd+(n-k)nd.
\end{align}
In the same way, we get that $\bE_{\calH_0}[\s{T}(\s{X},\s{Y})]=0$ and $\s{Var}_{\calH_0}\p{\s{T}(\s{X},\s{Y})} = (n-k)nd$. Recall that $\tau = \delta\cdot\rho kd$, for some $\delta<1$. Then,
\begin{align}
    \pr_{\calH_1}\pp{\phi(\s{X},\s{Y})=0} &\leq\frac{\s{Var}_{\calH_1}\p{\s{T}(\s{X},\s{Y})}}{\p{\bE_{\calH_1}[\s{T}(\s{X},\s{Y})]-\tau}^2},\\
    & = \frac{\rho^2\p{dk^2+dkn}+(1-\rho^2)knd+(n-k)nd}{(1-\delta)^2\rho^2k^2d^2}\\
    & = \frac{1}{(1-\delta)^2d}+\frac{n}{(1-\delta)^2dk}+\frac{(1-\rho^2)n}{(1-\delta)^2\rho^2kd}+\frac{(n-k)n}{(1-\delta)^2\rho^2k^2d}.
\end{align}
Thus, we see that if $d\to\infty$, $\rho^2d\frac{k^2}{n^2}\to\infty$, $\rho^2d\frac{k}{n}\to\infty$,  and $\frac{dk}{n}\to\infty$, then we have $\pr_{\calH_1}\pp{\phi(\s{X},\s{Y})=0}\to0$. Note however that $\rho^2d\frac{k^2}{n^2}\to\infty$ is more stringent compared to $\frac{dk}{n}\to\infty$ and $\rho^2d\frac{k}{n}\to\infty$; therefore, $\pr_{\calH_1}\pp{\phi(\s{X},\s{Y})=0}\to0$ if $d\to\infty$ and $\rho^2d\frac{k^2}{n^2}\to\infty$. As for the Type-I error probability, we get that
\begin{align}
    \pr_{\calH_0}\pp{\phi(\s{X},\s{Y})=1} &\leq\frac{\s{Var}_{\calH_0}\p{\s{T}(\s{X},\s{Y})}}{\p{\tau-\bE_{\calH_0}[\s{T}(\s{X},\s{Y})]}^2} \\
    &=\frac{(n-k)nd}{\delta^2\rho^2k^2d^2}\\
    &\leq \frac{n^2}{\delta^2\rho^2k^2d},
\end{align}
which again vanishes provided that $\rho^2d\frac{k^2}{n^2}\to\infty$. Therefore, we can conclude that the sum of the Type-I and Type-II error probabilities vanishes, as $d\to\infty$, if $\rho^2d\frac{k^2}{n^2}\to\infty$, as claimed.

\vspace{0.2cm} \noindent\textbf{Comparison test.} We now prove Theorem~\ref{thm:upper_comp_high}. We analyze the case where $\rho\in(0,1]$, with the understanding that the case where $\rho\in[-1,0)$ is analyzed in the same way. Let $G_1\triangleq\sum_{i=1}^n\sum_{j=1}^dX_{ij}$ and $G_2\triangleq\sum_{i=1}^n\sum_{j=1}^dY_{ij}$. Then, under $\calH_0$, we clearly have $G_1-G_2\sim\calN(0,2nd)$, while under $\calH_1$, we have $G_1-G_2\sim\calN(0,2nd-2kd\rho))$. Therefore, 
    \begin{align}
        1-\s{R}(\phi_{\s{comp}}) &= \pr_{\calH_0}(|G_1-G_2|\geq\theta)-\pr_{\calH_1}(|G_1-G_2|\geq\theta)\\
        & = \pr(|\calN(0,2nd)|\geq\theta)-\pr(|\calN(0,2nd-2kd\rho)|\geq\theta)\\
        & = d_{\s{TV}}\p{\calN(0,1),\calN\p{0,1-\frac{k}{n}\rho}}\\
        &= \Omega(1),
    \end{align}
    where the third equality holds by the definition of $\theta$ in \eqref{def:thetaTV}, and the last equality is because $\rho = \Omega(1)$ and $k = \Theta(n)$. 

\section{Conclusion and Outlook}

In this paper, we analyzed the detection problem of deciding whether two given databases are correlated or not, under a Gaussian distributive model. With a particular focus on the canonical one-dimensional case, we derived thresholds at which optimal testing is information-theoretically impossible and possible. Our lower bounding technique is based on an orthogonal polynomial expansion of the likelihood function, which revealed interesting connections to integer partition functions. While the study of the database alignment is relatively new, there are many open questions going forward. We thereby mention several directions of particular interest for future research:
\begin{enumerate}
    \item Our polynomial expansion method proved very useful in our setting. Thus, it would be interesting to develop a framework which implements similar ideas for general families of inference problems with a hidden combinatorial structure, e.g., in planted subgraph problems.  
    \item The information-theoretic thresholds for detection impossibility and possibility in the Gaussian database alignment problem are fairly understood for $d\to\infty$. Nevertheless, the case where $d\geq 2$ and fixed, is not solved completely. While the best currently known strong detection lower bound is given by $\min(1/d,\rho^\star)<1$, state-of-the-art strong detection algorithms work only if $\rho^2$ approaches $1$ sufficiently fast. Closing this evident gap is an important step towards complete understanding of the database alignment problem.
 
\end{enumerate}

%%%%%%
%% Appendix:
%% If needed a single appendix is created by
%%
%\appendix
%%
%% If several appendices are needed, then the command
%%
% \appendices
%%
%% in combination with further \section commands can be used.
%%%%%%

\section*{Acknowledgment}

The work of D. Elimelech was supported by the ISRAEL SCIENCE FOUNDATION (grant No. 985/23). The work of W. Huleihel was supported by the ISRAEL SCIENCE FOUNDATION (grant No. 1734/21).

\bibliographystyle{ieeetr}
\bibliography{bibfile} 

\appendix

\subsection{Proof of Lemma~\ref{lem:simplecalc}}\label{app:logmom}
We compute $\psi_Q(\lambda)$ using its definition,
\begin{align}
    \psi_Q(\lambda)
    &= \log\bE_Q\pp{\frac{1}{(1-\rho^2)^{\lambda/2}}\exp\p{\lambda\frac{-(X^2+Y^2)\rho^2+2 X Y\rho}{2(1-\rho^2)}}}\\
    & = -\frac{\lambda}{2}\log(1-\rho^2)+\log\bE_Q\pp{\exp\p{\lambda\frac{-(X^2+Y^2)\rho^2+2 X Y\rho}{2(1-\rho^2)}}}.\label{eq:psicalc}
\end{align}
Using the well-known fact that the moment generating function of Gaussian random variable $W\sim\calN(\mu,\sigma^2)$ squared is given by,
\begin{align}\label{eq:GausSquare}
    \bE [\exp\p{tW^2}] = \frac{1}{\sqrt{1-2t\sigma^2}} \exp\left(\frac{\mu^2 t}{1-2t\sigma^2}\right),
\end{align}
when $\s{Real}(t\sigma^2)<1/2$, we have that for $-\frac{1-\rho^2}{\rho^2}\leq\lambda$ and $|1-\lambda|<\frac{1}{\rho}$,
\begin{align}
    \bE_Q\pp{\exp\p{\lambda\frac{-(X^2+Y^2)\rho^2+2 X Y\rho}{2(1-\rho^2)}}} 
    &=\bE_Q\pp{\E\pp{\exp\p{\lambda\frac{-(X^2+Y^2)\rho^2+2 X Y\rho}{2(1-\rho^2)}}~\bigg|~X}}
    \\&=\bE_Q\pp{e^{-\frac{\lambda\rho^2}{2(1-\rho^2)}X^2}\E\pp{e^{\p{\lambda\frac{-Y^2\rho^2+2 X Y\rho}{2(1-\rho^2)}}}~\bigg|~X}}
    \\&=\bE_Q\pp{e^{\frac{-\lambda(\rho^2-1)}{2(1-\rho^2)}X^2}\E\pp{e^{\p{-\lambda\frac{(\rho Y -X)^2}{2(1-\rho^2)}}}~\bigg|~X}}\\
    &\overset{(a)}{=}\bE_Q\pp{e^{\frac{-\lambda(\rho^2-1)}{2(1-\rho^2)}X^2}\sqrt{\frac{1+\lambda\rho^2}{1-\rho^2+\lambda\rho^2}} \cdot e^{\frac{-\lambda X^2}{2(1-\rho^2+\lambda\rho^2)}}}\\
    &=\sqrt{\frac{1-\rho^2}{1-\rho^2+\lambda\rho^2}}\bE_Q\pp{ e^{\frac{-\lambda(1-\lambda)\rho^2 }{2(1-\rho^2+\lambda\rho^2)}\cdot X^2}}\\
    &\overset{(b)}{=}\sqrt{\frac{1-\rho^2}{1-\rho^2+\lambda\rho^2}} \cdot \sqrt{\frac{1-\rho^2+\lambda\rho^2}{1-(1-\lambda)^2\rho^2}} \\
    &=\frac{\sqrt{1-\rho^2}}{\sqrt{1-(1-\lambda)^2\rho^2}}, \label{eq:uglyexponent}
\end{align}
%$-\frac{1-\rho^2}{\rho^2}\leq\lambda$ and $\lambda[1-(1-(1-\lambda)\rho^2)^{-1}]<1$ \textcolor{red}{
where $(a)$ follows from \eqref{eq:GausSquare} with $t=-\lambda/(2(1-\rho^2))$, $\mu=-X$ and $\sigma^2=\rho^2$ as given $X$, $\rho Y-X$ is distributed as $\calN(X,\rho^2)$.  $(b)$ follows from \eqref{eq:GausSquare} as well, with $\mu=0$, $\sigma^2=1$ and 
$t=\frac{-\lambda(1-\lambda)\rho^2 }{2(1-\rho^2+\lambda\rho^2)}$. Plugging in \eqref{eq:uglyexponent} to \eqref{eq:psicalc}, we conclude the first part of the proof.

The calculation of $\psi_P(\lambda)$ is very similar. As in \eqref{eq:uglyexponent} we have 
\begin{align}
    &\psi_P(\lambda) = -\frac{\lambda}{2}\log(1-\rho^2)+\log\bE_P\pp{\exp\p{\lambda\frac{-(X^2+Y^2)\rho^2+2 X Y\rho}{2(1-\rho^2)}}}.\label{eq:psicalc2}
\end{align}
We repeat the the calculation of the moment generating function using \eqref{eq:GausSquare} and obtain that for all $\abs{\lambda}\leq 1/\rho^2$ we have 
\begin{align}
    \bE_P\pp{\exp\p{\lambda\frac{-(X^2+Y^2)\rho^2+2 X Y\rho}{2(1-\rho^2)}}} 
    &=\bE_P\pp{\E\pp{\exp\p{\lambda\frac{-(X^2+Y^2)\rho^2+2 X Y\rho}{2(1-\rho^2)}}~\bigg|~X}}
    \\&=\bE_P\pp{e^{-\frac{\lambda\rho^2}{2(1-\rho^2)}X^2}\E\pp{e^{\p{\lambda\frac{-Y^2\rho^2+2 X Y\rho}{2(1-\rho^2)}}}~\bigg|~X}}
    \\&=\bE_P\pp{e^{\frac{-\lambda(\rho^2-1)}{2(1-\rho^2)}X^2}\E\pp{e^{\p{-\lambda\frac{(\rho Y -X)^2}{2(1-\rho^2)}}}~\bigg|~X}}\\
    &\overset{(a)}{=}\bE_P\pp{e^{\frac{-\lambda(\rho^2-1)}{2(1-\rho^2)}X^2}\sqrt{\frac{1}{1+\lambda\rho^2}} \cdot e^{\frac{-\lambda(1-\rho^2) X^2}{2(1+\lambda\rho^2)}}}\\
    &=\sqrt{\frac{1}{1+\lambda\rho^2}}\bE_P\pp{ e^{\frac{\lambda\rho^2(1+\lambda) }{2(1+\lambda \rho^2)}\cdot X^2}}\\
  &=\sqrt{\frac{1}{1+\lambda\rho^2}}\cdot \sqrt{\frac{1+\lambda \rho^2}{1-\rho^2}}\\
  &=\frac{1}{\sqrt{1-\rho^2}},\label{eq:uglyexponent2}
\end{align}
where $(a)$ follows from \eqref{eq:GausSquare} with $t=-\lambda/(2(1-\rho^2))$, $\mu=X(\rho^2-1)$ and $\sigma^2=\rho^2(1-\rho^2)$ as given $X$, under the measure $P$, $\rho Y-X$,  is distributed as $\calN(X(1-\rho^2),\rho^2(1-\rho^2))$. $(b)$ follows from \eqref{eq:GausSquare} as well, with $\mu=0$, $\sigma^2=1$ and 
$t=\frac{\lambda\rho^2(1+\lambda) }{2(1+\lambda \rho^2)}$. 
Plugging in \eqref{eq:uglyexponent2} to \eqref{eq:psicalc2}, we conclude.

\subsection{Proof of Proposition~\ref{prop:nastybound}}\label{app:cyclesPer}
    We follow the steps of the proof of \cite[Theorem 2]{nazer2022detecting}. Let $\calK\sim\s{Unif}\binom{[n]}{k}$, and $\sigma$ be chosen uniformly over $\S_n$, where $\sigma$ and $\calK$ are mutually independent. By the model definition, it is easy to see that 
\begin{align}
    \calL(\s{X},\s{Y}) & = \frac{\pr_{\calH_1}(\s{X},\s{Y})}{\pr_{\calH_0}(\s{X},\s{Y})}\\
    & = \bE_{\sigma}\bE_{\calK}\pp{\prod_{i\in\calK}\frac{\pr_{\calH_1}(X_i,Y_{\sigma(i)})}{\pr_{\calH_0}(X_i,Y_{\sigma(i)})}}.
\end{align}
Let $\calK'$ and $\sigma'$ be independent copies of $\calK$ and $\sigma$, respectively. Then, using Fubini's theorem we have,
\begin{align}
    \calL(\s{X},\s{Y})^2&=\bE_{\substack{\sigma\indep\sigma'\\\calK\indep\calK'}}\pp{\prod_{i\in\calK}\frac{\pr_{\calH_1}(X_i,Y_{\sigma(i)})}{\pr_{\calH_0}(X_i,Y_{\sigma(i)})}\prod_{i\in\calK'}\frac{\pr_{\calH_1}(X_i,Y_{\sigma'_i})}{\pr_{\calH_0}(X_i,Y_{\sigma'_i})}}\\
    & = \bE_{\substack{\sigma\indep\sigma'\\\calK\indep\calK'}}\bigg[\prod_{i\in\calK\cap\calK'}\frac{\pr_{\calH_1}(X_i,Y_{\sigma(i)})}{\pr_{\calH_0}(X_i,Y_{\sigma(i)})}\frac{\pr_{\calH_1}(X_i,Y_{\sigma'_i})}{\pr_{\calH_0}(X_i,Y_{\sigma'_i})}\nonumber\\
    &\hspace{1.8cm} \cdot\prod_{i\in\calK\setminus\calK'}\frac{\pr_{\calH_1}(X_i,Y_{\sigma(i)})}{\pr_{\calH_0}(X_i,Y_{\sigma(i)})}\prod_{i\in\calK'\setminus\calK}\frac{\pr_{\calH_1}(X_i,Y_{\sigma'_i})}{\pr_{\calH_0}(X_i,Y_{\sigma'_i})} \bigg]\label{eqn:secondmom}\\
    & \triangleq \bE_{\sigma\indep\sigma'}\bE_{\calK\indep\calK'}\pp{g(\calK,\calK',\sigma,\sigma')}.
\end{align}
Thus, the second moment of the likelihood is given by,
\begin{align}
    \bE_{\calH_0}[\calL(\s{X},\s{Y})^2] = \bE_{\sigma\indep\sigma'}\bE_{\calK\indep\calK'}\bE_{\calH_0}\pp{g(\calK,\calK',\sigma,\sigma')}.
\end{align}
Looking at \eqref{eqn:secondmom} we note that the expectation of the last two products w.r.t. $\calH_0$ is unity. Specifically, for simplicity of notation, define $g_{\calA}(\sigma)\triangleq\prod_{i\in\calA}\frac{\pr_{\calH_1}(X_i,Y_{\sigma(i)})}{\pr_{\calH_0}(X_i,Y_{\sigma(i)})}$. Then, we have
\begin{align}
  \bE_{\calH_0}\pp{g(\calK,\calK',\sigma,\sigma')}
  &= \bE_{\calH_0}\bigg[g_{\calK\cap\calK'}(\sigma)g_{\calK\cap\calK'}(\sigma')g_{\calK\setminus\calK'}(\sigma)g_{\calK'\setminus\calK}(\sigma')\bigg]\nonumber\\
    &=\bE_{\calH_0}\bigg[\bE_{\calH_0}[g_{\calK\cap\calK'}(\sigma)g_{\calK\cap\calK'}(\sigma')g_{\calK\setminus\calK'}(\sigma)g_{\calK'\setminus\calK}(\sigma') \nonumber\\
    &   \hspace{2cm}~\big|~ \s{X}_{\calK\cap\calK'},\s{Y}_{\sigma(\calK\cap\calK')},\s{Y}_{\sigma'(\calK\cap\calK')}]\bigg] \\
    &=\bE_{\calH_0}\bigg[ g_{\calK\cap\calK'}(\sigma)g_{\calK\cap\calK'}(\sigma')\bE_{\calH_0}[g_{\calK\setminus\calK'}(\sigma)g_{\calK'\setminus\calK}(\sigma') \nonumber\\
    &  \hspace{2cm}~\big|~ \s{X}_{\calK\cap\calK'},\s{Y}_{\sigma(\calK\cap\calK')},\s{Y}_{\sigma'(\calK\cap\calK')}]\bigg]\\
        &\overset{(a)}{=}\bE_{\calH_0}\bigg[ g_{\calK\cap\calK'}(\sigma)g_{\calK\cap\calK'}(\sigma')\nonumber \\
    &   \hspace{2cm}\cdot\bE_{\calH_0}[g_{\calK\setminus\calK'}(\sigma)g_{\calK'\setminus\calK}(\sigma')~\big|~ \s{Y}_{\sigma(\calK\cap\calK')},\s{Y}_{\sigma'(\calK\cap\calK')}]\bigg] \\
        &\overset{(b)}{=}\bE_{\calH_0}\bigg[ g_{\calK\cap\calK'}(\sigma)g_{\calK\cap\calK'}(\sigma')\cdot\bE_{\calH_0}\pp{\left.g_{\calK\setminus\calK'}(\sigma)  \right| \s{Y}_{\sigma(\calK\cap\calK')},\s{Y}_{\sigma'(\calK\cap\calK')}}\nonumber \\
        &\hspace{1.5cm} \cdot\bE_{\calH_0}\pp{\left.g_{\calK'\setminus\calK}(\sigma')\right| \s{Y}_{\sigma(\calK\cap\calK')},\s{Y}_{\sigma'(\calK\cap\calK')}}\bigg]\\
    &\overset{(c)}{=}\bE_{\calH_0}\bigg[g_{\calK\cap\calK'}(\sigma)g_{\calK\cap\calK'}(\sigma')\bE_{\calH_0}\pp{\left.g_{\calK\setminus\calK'}(\sigma) \right| \s{Y}_{\sigma(\calK\cap\calK')}
    }\nonumber\\
    &\hspace{1.6cm}\cdot\bE_{\calH_0}\pp{\left.g_{\calK'\setminus\calK}(\sigma') \right|\s{Y}_{\sigma'(\calK\cap\calK')}}\bigg] \\
    &= \bE_{\calH_0}\bigg[ g_{\calK\cap\calK'}(\sigma)g_{\calK\cap\calK'}(\sigma')\bE_{\s{X}_{\calK\setminus\calK'}}[g_{\calK\setminus\calK'}(\sigma)]\bE_{\s{X}_{\calK'\setminus\calK}}[g_{\calK\setminus\calK'}(\sigma)]\bigg],\label{eqn:innerExpexindep}
\end{align}
where $(a)$ follows from the independence of $(\s{X}_{\calK\triangle \calK'}, \s{Y}_{\sigma(\calK\cap\calK')},\s{Y}_{\sigma'(\calK\cap\calK')})$ on $\s{X}_{\calK\cap\calK'}$, $(b)$ and $(c)$ follow similarly by independence properties of $\s{X}$ and $\s{Y}$, and the inner expectation $\bE_{\s{X}_{\calK\setminus\calK'}}$ in \eqref{eqn:innerExpexindep} is taken w.r.t. the distribution of $\{X_i\}_{i\in\calK\setminus\calK'}$, under $\calH_0$, and similarly for $\bE_{\s{X}_{\calK'\setminus\calK}}$, while the outer expectation in \eqref{eqn:innerExpexindep} is over the remaining random variables. Now, we note that
\begin{align}
    \bE_{\s{X}_{\calK\setminus\calK'}}[g_{\calK\setminus\calK'}(\sigma)]
    &= \bE_{\s{X}_{\calK\setminus\calK'}}\pp{\prod_{i\in\calK\setminus\calK'}\frac{\pr_{\calH_1}(X_i,Y_{\sigma(i)})}{\pr_{\calH_0}(X_i,Y_{\sigma(i)})}}\\
    & = \int_{\s{X}_{\calK\setminus\calK'}}\prod_{i\in\calK\setminus\calK'}\pr_{\calH_0}(X_i)\prod_{i\in\calK\setminus\calK'}\frac{\pr_{\calH_1}(X_i,Y_{\sigma(i)})}{\pr_{\calH_0}(X_i,Y_{\sigma(i)})}\mathrm{d}\s{X}_{\calK\setminus\calK'}\\
    & = \int_{\s{X}_{\calK\setminus\calK'}}\prod_{i\in\calK\setminus\calK'}\frac{\pr_{\calH_1}(X_i,Y_{\sigma(i)})}{\pr_{\calH_0}(Y_{\sigma(i)})}\mathrm{d}\s{X}_{\calK\setminus\calK'}\\
    & = \prod_{i\in\calK\setminus\calK'}\frac{\pr_{\calH_1}(Y_{\sigma(i)})}{\pr_{\calH_0}(Y_{\sigma(i)})}\\
    & = 1.
\end{align}
In the same way, we have $\bE_{\s{X}_{\calK'\setminus\calK}}[g_{\calK'\setminus\calK}(\sigma)]=1$. Therefore,
\begin{align}
    \bE_{\calH_0}[\calL(\s{X},\s{Y})^2]
    &= \bE_{\substack{\sigma\indep\sigma'\\ \calK\indep\calK'}}\bigg[\bE_{\calH_0}\pp{g_{\calK\cap\calK'}(\sigma)g_{\calK\cap\calK'}(\sigma')}\bigg]\\
    & = \bE_{\substack{\sigma\indep\sigma'\\ \calK\indep\calK'}}\pp{\bE_{\calH_0}\pp{\prod_{i\in\calK\cap\calK'}\frac{\pr_{\calH_1}(X_i,Y_{\sigma(i)})}{\pr_{\calH_0}(X_i,Y_{\sigma(i)})}\frac{\pr_{\calH_1}(X_i,Y_{\sigma'_i})}{\pr_{\calH_0}(X_i,Y_{\sigma'_i})}}}.
\end{align}
By symmetry we may assume that $\pi' = \s{Id}$ and $\calK'=[k]$. Then, the expression above simplifies to,
\begin{align}
    \bE_{\calH_0}[\calL(\s{X},\s{Y})^2]
    &= \bE_{\sigma,\calK}\pp{\bE_{\calH_0}\pp{\prod_{i\in[k]\cap\calK}\frac{\pr_{\calH_1}(X_i,Y_{i})}{\pr_{\calH_0}(X_i,Y_{i})}\frac{\pr_{\calH_1}(X_i,Y_{\sigma(i)})}{\pr_{\calH_0}(X_i,Y_{\sigma(i)})}}}.\label{eq:LikeRatDis}
\end{align}
Denote the product in \eqref{eq:LikeRatDis} by $Z_{\sigma,\calK}$. We recall that a cycle of a permutation $\sigma$ is a string $(i_0,i_2,\dots,i_{|C|-1})$ of elements in $[n]$ such that $\sigma(i_j)=i_{j+1\mod|C|}$  for all $j$. If $\abs{C}=\ell$, we call $C$ a $\ell$-cycle. For a fixed cycle $C$, we denote
\begin{align}
Z_{C,\calK}\triangleq \prod_{i\in[k]\cap\calK\cap C}\frac{\pr_{\calH_1}(X_i,Y_{i})}{\pr_{\calH_0}(X_i,Y_{i})}\frac{\pr_{\calH_1}(X_i,Y_{\sigma(i)})}{\pr_{\calH_0}(X_i,Y_{\sigma(i)})}.
\end{align}
Since the set of cycles of a permutation induce a partition of $[n]$, the random variables $\{Z_{C,\calK}\}_C$, corresponding to all cycles of $\sigma$, are independent (w.r.t. $\P_{\calH_0}$), and thus,
\begin{equation}
    \label{eq:CycleProd}
    Z_{\sigma,\calK}=\prod_{C}Z_{C,\calK}.
\end{equation} 
The following lemma is a straightforward generalization that follows easily from the proof of \cite[Lemma~10]{nazer2022detecting}.
\begin{lemma}\label{lem:CycleComp}
For a fixed cycle $C$ of a permutation $\sigma$, and a set $\calK$, we have
\begin{align}
\E_{\calH_0}[Z_{C,\calK}]=\frac{1}{(1-\rho^{2|[k]\cap\calK\cap C|})^{d}}.
\end{align}
\end{lemma}
Using this lemma, we get that
\begin{align}
    \bE_{\calH_0}[\calL(\s{X},\s{Y})^2] &= \bE_{\sigma}\bE_{\calK}\bE_{\calH_0}\pp{Z_{\sigma,\calK}}\\
    & = \bE_{\sigma}\bE_{\calK}\bE_{\calH_0}\pp{\prod_{C}Z_{C,\calK}}\\
    & = \bE_{\sigma}\bE_{\calK}\pp{\prod_{C}\bE_{\calH_0}\pp{Z_{C,\calK}}}\\
    & = \bE_{\sigma}\bE_{\calK}\pp{\prod_{C}\frac{1}{(1-\rho^{2|[k]\cap\calK\cap C|})^{d}}}\\
    & = \bE_{\sigma}\bE_{\calK}\pp{\prod_{\ell=1}^k\frac{1}{(1-\rho^{2\ell})^{d\s{N}_{\ell}(\sigma,\calK)}}},\label{eq:133}
\end{align}
where $\s{N}_{\ell}(\sigma,\calK)$ is defined as
\begin{align}
    \s{N}_{\ell}(\sigma,\calK) = \abs{\ppp{C\in\sigma:\;|[k]\cap\calK\cap C|=\ell}}.
\end{align}  

%%%%%%
%% To balance the columns at the last page of the paper use this
%% command:
%%
%\enlargethispage{-1.2cm} 
%%
%% If the balancing should occur in the middle of the references, use
%% the following trigger:
%%

%% which triggers a \newpage (i.e., new column) just before the given
%% reference number. Note that you need to adapt this if you modify
%% the paper.  The "triggered" command can be changed if desired:
%%
%\IEEEtriggercmd{\enlargethispage{-20cm}}
%%
%%%%%%

%%%%%%
%% References:
%% We recommend the usage of BibTeX:
%%
%\bibliographystyle{IEEEtran}
%\bibliography{definitions,bibliofile}
%%
%% where we here have assumed the existence of the files
%% definitions.bib and bibliofile.bib.
%% BibTeX documentation can be obtained at:
%% http://www.ctan.org/tex-archive/biblio/bibtex/contrib/doc/
%%%%%%

%% Or you use manual references (pay attention to consistency and the
%% formatting style!):

\end{document}